\documentclass[conference]{IEEEtran}
\usepackage{cite}

\ifCLASSINFOpdf
\else
\fi
\usepackage[bookmarks=false]{hyperref}
\usepackage{graphicx}
\usepackage{amssymb}
\usepackage{amsmath}
\usepackage{epstopdf}
\usepackage{algorithm}
\usepackage[noend]{algorithmic}
\usepackage{color}
\usepackage{paralist}
\usepackage{url}
\usepackage{caption}
\usepackage{subcaption}
\usepackage{balance}

\usepackage{amsthm}

\newtheorem{lemma}{Lemma}

\newcommand{\graph}{{\cal G}}
\newcommand{\vertices}{{\cal V}}
\newcommand{\edges}{{\cal E}}
\newcommand{\samp}[1]{{#1}_{S}}

\newcommand{\alg}{\textsc{EWSample}}
\newcommand{\algws}{\textsc{WedgeSample}}
\newcommand{\alges}{\textsc{EdgeSample}}

\newcommand{\lowdeg}[1]{\delta(#1)}

\begin{document}
\begin{NoHyper}

\title{Edge-Based Wedge Sampling to Estimate Triangle Counts in Very Large Graphs}

\author{\IEEEauthorblockN{Duru T\"urko\u{g}lu}
\IEEEauthorblockA{DePaul University, Chicago, IL\\
\texttt{dturkogl@cs.depaul.edu}}
\and
\IEEEauthorblockN{Ata Turk}
\IEEEauthorblockA{Boston University, Boston, MA\\
\texttt{ataturk@bu.edu}}
}

\maketitle

\sloppy

\begin{abstract}

The number of triangles in a graph is useful to deduce a plethora of
important features of the network that the graph is modeling.
However, finding the exact value of this number is computationally expensive. 
Hence, a number of approximation algorithms based on random 
sampling of edges, or wedges (adjacent edge pairs) have been proposed for estimating
this value.
%
%
%
We argue that for large sparse graphs with power-law degree
distribution, random edge sampling requires sampling large number of
edges before providing enough information for accurate estimation, and existing wedge sampling methods lead to biased samplings, which in turn lead to less accurate estimations.
In this paper, we propose a hybrid algorithm between edge and wedge
sampling that addresses the deficiencies of both approaches.
We start with uniform edge sampling and then extend each selected edge
to form a wedge that is more informative for estimating the overall
triangle count.
The core estimate we make is the number of triangles each sampled edge
in the first phase participates in.
This approach provides accurate approximations with very small
sampling ratios, outperforming the state-of-the-art up to 8 times in
sample size while providing estimations with 95\% confidence.

\end{abstract}




\section{Introduction}
\label{sec:introduction}

Graphs are heavily employed in modeling relationships, networks, and interactions of many real world applications. User interactions and usage patterns in social or mobile networks, compound/atom/amino-acid/molecule bonding/interaction tendencies in chemistry or biology, process/data/resource interaction/communication/dependency patterns in parallel/distributed systems can be listed among the examples where graph-based modeling is heavily used.  

Triangles in a graph represent ternary relationships among the modeled
objects/entities. The count of triangles in a
graph is an important metric that can be used in determining the degree of
clustering of the modeled networks. Graph metrics such as clustering
coefficient~\cite{Newman:2003} or transitivity ratio~\cite{Wasserman:1994} make
use of triangle count, and they are used as features for critical
applications such as social network analysis~\cite{Yang:2014}, gene expression
microarray data analysis~\cite{Kalna:2007}, or text
summarization~\cite{Al:2014}.

Currently, the best known theoretical algorithm for identifying the exact number of triangles in a given graph $\graph$ has a complexity of $O\left(m^{2\gamma/(\gamma+1)}\right) \approx O(m^{1.407})$~\cite{Alon:1997}, where
$m$ is the number of edges in $\graph$, and $\gamma < 2.372864$. This
algorithm computes the third power of the adjacency matrix on the
high-degree vertices and $\gamma$ is the exponent of the state-of-the-art matrix multiplication algorithm~\cite{LeGall:2014}.
The practically used triangle counting and listing algorithms,
however,
take $O(m^{3/2})$ time~\cite{Schank:2005a}.
Moreover, since finding the exact triangle count is computationally
expensive for very large graphs with billions of edges that do not fit
in memory, some studies focus on minimizing the number of disk I/Os performed~\cite{Hu:2014, Kim:2016} during counting. 

In many scenarios finding the exact triangle count is not necessary
and an approximate number is sufficient as the graphs are very
dynamic. A number of studies propose
approximation solutions for estimating the triangle count~\cite{Bar-Yossef:2002, Jowhari:2005, Buriol:2006, Latapy:2008, Tsourakakis:2009, Jha:2013, Schank:2005a, Seshadhri:2014, Etemadi:2016}. These studies can be categorized based on the settings they consider: the streaming setting where graph components arrive as a data stream~\cite{Bar-Yossef:2002, Buriol:2006, Jowhari:2005, Jha:2013}, the semi-streaming setting where a constant number of passes over the edges are allowed~\cite{Becchetti:2008, Kolountzakis:2010}, and the static setting where the entire graph is available for analysis~\cite{Tsourakakis:2009, Seshadhri:2013, Seshadhri:2014}.
Another vein of studies consider efficient parallelization of this estimation process~\cite{Suri:2011, Pagh:2012, Kolda:2014}.
In this paper, we focus on the triangle count estimation problem under the static sequential setting. However, the intuition we provide can be used in a streaming setting as well.

The state-of-the-art approximation approaches are mostly based on random edge or wedge sampling. 
Random wedge sampling approaches first sample random vertices where vertex sampling probability is proportional to wedge participation probability of that vertex. After vertex sampling, two edges of each selected vertex are sampled to form wedges. The ratio of sampled wedges to closed wedges is used to estimate the number of triangles in
the original graph~\cite{Seshadhri:2014}. When deployed on power-law degree graphs, wedge sampling approaches can suffer from bias (mostly high degree vertices are selected as hinge points of selected wedges), as they start with vertex sampling.
Random edge sampling approaches count the number of triangles~\cite{Tsourakakis:2009} or the number of wedges that are ``closed''~\cite{Etemadi:2016} (that have a third closing edge in the original graph) in the sampled subgraph, and use this count to estimate the triangle count. However, if graphs are sparse and the sampling ratio is low, these approaches can suffer from scarcity of valuable information (e.g. triangles, wedges) in the sampled subgraph.

We propose an algorithm for the graph triangle count estimation
problem that combines the strong points of edge and wedge sampling
while alleviating their deficiencies. Our approach starts with random
edge sampling to avoid challenges arising in power-law graphs. Then we
first estimate the number of triangles the sampled edges participate
in the original graph. We use that first estimate to estimate the
triangle count of the original graph. 
We accomplish this 
by turning the sampled
edges into wedges in a second random edge selection phase and by checking if these sampled wedges are closed.

Our contributions can be listed as follows: (i) We propose an approach
that offers highly accurate triangle count estimations even when
employed over sparse power-law degree graphs and the sampling ratio is
low. (ii) We theoretically prove the bounds of our algorithm and show
with experiments that found bounds are tight. We also provide
theoretical bounds for the state-of-the-art approaches. (iii) We
theoretically and practically show that the proposed approach achieves up
to eight times sampling size reduction over the state-of-the-art 
on large-scale real-world graphs.

The remainder of this paper is organized as follows. We present
existing approximation approaches in Section~\ref{sec:background}
along with our observations that led to the development of the proposed
algorithm. Section~\ref{sec:proposed_method} discusses details of our
proposed edge-based wedge selection algorithm. We provide a
theoretical analysis of the proposed method and other methods in the
literature in Section~\ref{sec:rse}. We evaluate the  performance of
the proposed method and compare it against the state-of-the-art in
Section~\ref{sec:eval}. In Section~\ref{sec:conclusion} we conclude.


\vspace{-1ex}
\section{Counting Triangles in a Graph}
\label{sec:background}

%

Given an undirected graph $\graph=(\vertices, \edges)$ defined by a set of 
vertices ${\cal V}$ and a set $\edges$ of pairwise relations (edges) among vertices in ${\cal V}$, we would like to estimate $\Delta$, the number of triangles (three edge cycles) in $\graph$. 
Approximation approaches to this problem mostly consider randomly
sampling edges or adjacent edge pairs (wedges) from the original graph $\graph$ and try to estimate the number of triangles in $\graph$ based on analysis conducted on the sampled graph entities.

%

Existing approximation methods have certain deficiencies
when applied to graphs modeling real networks that exhibit power-law
degree distributions with many low-degree vertices and few very
high-degree vertices~\cite{Barabasi:1999}. Using the current random edge
sampling methods,
sampled subgraphs can be too sparse
to reveal sufficient information for accurate estimation when low
sampling ratios are employed.  Using the current random wedge-sampling
methods, the uniformity on the wedge selection scheme favors
high-degree vertices, and this leads to less accurate estimations.   

In the remainder of this section, we explain the state-of-the-art edge and
wedge sampling based approaches in detail, discuss their strengths and
deficiencies, and provide our intuitions to address these
deficiencies.  In Table~\ref{tbl:notations} we summarize  the
important notations that we use in our discussions and in the
remainder of the paper.

\vspace{-1ex}
\subsection{Random edge sampling }
Random edge sampling based approaches first construct a subgraph
$\samp{\graph}$ of $\graph$ by traversing over each edge in $\edges$
and adding each edge to the subgraph $\samp{\graph}$ with probability
$p$. They then extract features from $\samp{\graph}$ 
for estimating $\Delta$.

\subsubsection{Counting triangles in $\samp{\graph}$}
The idea of random edge sampling to construct a subgraph $\samp{\graph}$ was
first proposed by Tsourakakis et al.~\cite{Tsourakakis:2009}. In their
work, they use
the number of triangles in $\samp{\graph}$ to estimate $\Delta$ as
$\samp{\Delta}/p^3$, where $\samp{\Delta} = 
\left| \{ (u,v,w) ~|~ (u,v),(v,w),(w,u) \in \samp{\graph} \} \right|$. This
estimation is based on the fact that all three edges of any triangle
in $\samp{\graph}$ has to be selected with probability $p$ from $\graph$. A
drawback of this approach is that when $p$ is small, $\samp{\graph}$ is very
unlikely to contain even one triangle, making this approach
impractical for small sampling ratios.

\subsubsection{Counting closed wedges in $\samp{\graph}$}
Recently, Etemadi et al.~\cite{Etemadi:2016} proposed an approach that
improves upon \cite{Tsourakakis:2009}. They first identify all of the
wedges formed in $\samp{\graph}$, and then check for edges in $\graph$
to determine whether the wedges in $\samp{\graph}$ are closed in $\graph$. That is, for
any wedge $u-$$v$$-w$ in $\samp{\graph}$, they check whether the third
edge $(w,u)$ is in~$\graph$ or not.  They use the count of such \textit{closed wedges}
to estimate $\Delta$ as $\samp{\Lambda}^+ / 3p^2$, where
$\samp{\Lambda}^+$=$\left| \{ (u,v,w) ~|~ (u,v),(v,w) \in
\samp{\graph}, (w,u) \in \graph \} \right|$.\footnote{
The original definition in~\cite{Etemadi:2016} has the $\frac{1}{3}$
multiplier but we instead count the actual number of closed wedges and
report one third to estimate $\Delta$.}
This estimation is due to the following two facts that the two
edges forming a wedge in $\samp{\graph}$ have to be independently selected
with probability $p$, and that any triangle contains three wedges. 

\begin{table}
\begin{tabular}{c l}
\hline
Symbol & Description\\
\hline
$\graph$ & Original graph with vertex set $\vertices$ and edge set $\edges$\\
$\Delta$ & Number of triangles in $\graph$\\
$\Lambda$ & Number of wedges in $\graph$ \\
$\samp{\graph}$ & Random subgraph by edge selection\\
$\samp{\Delta}$ & Number of triangles in $\samp{\graph}$\\
$\samp{\Lambda}^{+}$ & Number of closed wedges in $\samp{\graph}$ \\
$T(e)$ & Number of triangles in $\graph$ that contain edge $e$ \\
$\samp{T}$ & Sum of $T(e)$ for edges in $\samp{\graph}$: $\sum_{e \in \samp{\graph}} T(e)$\\
$\tau$ & Estimate on $\samp{T}$ ({\small \alg})\\
$d_u$ & Degree of vertex $u$ \\
$\delta(e)$ & Degree of lower degree end vertex of edge $e$=$(u,v)$: \\ 
            & $\delta(e)$=$min(d_u, d_v)$\\
$\sigma(t)$ & Sum of $\delta(e)$ for edges in triangle $t$: $\sum_{e \in t} \delta(e)$\\
\hline
\end{tabular}
\caption{Summary of the notations.}
\label{tbl:notations}
\vspace{-3.5ex}
\end{table}

Although~\cite{Etemadi:2016} enables selection of smaller
$p$ values compared to the proposed approach in~\cite{Tsourakakis:2009}, we
note that 
achieving accurate estimates still requires
a large sampling probability $p$ as the chance of observing
sufficient number of wedges remains low when sampling edges from large
sparse graphs. In particular, for a graph $\graph$ with $m$ edges and $\Delta$ triangles, the expected number of wedges in the sampled graph that is part of a triangle in the original graph is $3\Delta p^2$~\cite{Etemadi:2016}
%
.  
In our subsequent experiments, $\Delta$ can be in the order of
$10^{10}$, and we would like to support $p$ values in the range
$[10^{-5}, 10^{-6}]$. In such a scenario, clearly, the number of closed
wedges in the sampled subgraph (each of which form a triangle in the
original graph) will be too low to provide accurate estimations.  

\subsubsection{Strengths of edge-based sampling approaches}
Random edge-based sampling methods do not require any previous
knowledge on the structure of $\graph$ and can be applied without any
preprocessing. Furthermore, in power-low graphs, factoring
computations over edges instead of vertices can alleviate challenges
associated with the power-law degree distribution of vertices and can also ease parallelism~\cite{Gonzalez:2012}.

\subsubsection{Deficiencies of edge-based sampling approaches}
The core problem of existing edge-based sparsification approaches is that
the relationship between the findings they extract from the sampled
graph $\samp{\graph}$ and the estimation they make for triangle count
in $\graph$ is a cubic or quadratic function of $p$. This prevents
accurate estimation for low $p$ values. Ideally, the estimation should
be directly proportional to $p$. 

\subsubsection{Proposed improvement over edge sampling}
The intuition behind our proposed approach is similar in nature
to the improvement proposed by~\cite{Etemadi:2016}
over~\cite{Tsourakakis:2009}. In~\cite{Etemadi:2016}, they point out that instead of counting
entities that appear with probability $p^3$ (triangles in
$\samp{\graph}$), they can count entities that appear with probability
$p^2$ (closed wedges in $\samp{\graph}$) and thus enable
using lower sampling ratios while achieving the same 
accuracy.
In this work, we take one more step further and estimate $\Delta$
using entities that appear with probability $p$.

In our approach, we also start with randomly selecting a set of edges $\samp{\edges}\subset\edges$. 
Then, using each edge $e$ in $\samp{\edges}$,
we directly estimate the number of triangles that $e$ is a part of.
More formally, let $T(e)$ be the number of triangles in $\graph$ that
$e$ is a part of. 
Furthermore, let $\samp{T}$ be the total number of triangles in
$\graph$ that each
edge in $\samp{\edges}$ is a part of, i.e.
$\samp{T}=\sum_{e\in\samp{\edges}} T(e)$.  In our approach, we
estimate $\samp{T}$, which can also be written as:
$\samp{T}=
\left| \{ (u,v,w) ~|~ (u,v) \in \samp{\edges},
(v,w), (w,u) \in \graph \} \right|$. 
We call our estimation for $\samp{T}$ as $\tau$ and further
estimate $\Delta$ as $\tau/3p$, since an edge can be selected with
probability $p$ and since each triangle has three edges.
Our estimate of $\samp{T}$ is directly proportional to
$p$.  We explain in detail how we estimate $\samp{T}$ and why the probability
associated with $\samp{T}$ estimation is proportional to $p$ in
Section~\ref{sec:proposed_method} and Section~\ref{sec:rse} respectively.

%

\subsection{Random wedge sampling}
Random wedge sampling approaches extract wedges 
 from $\graph$ and analyze properties of extracted wedges to estimate $\Delta$. 

\subsubsection{Uniform wedge sampling}
Random uniform wedge sampling for triangle counting was first proposed by Schank and Wagner~\cite{Schank:2005}.  
More recently, Seshadhri et al.~\cite{Seshadhri:2013, Seshadhri:2014}
analyze this approach in depth for counting various triadic
measures, including counting triangles. Similar appraoches have also
been considered in a streaming setting~\cite{Jha:2013}. 

In a nutshell,
uniform random wedge sampling first picks a vertex $v$ in $\vertices$
based on the number of wedges hinging at $v$  and then investigates if
a randomly created wedge that consists of two randomly selected edges
of $v$ is a closed wedge or not. The ratio of closed wedges to sampled
wedges is used to estimate $\Delta$ by multiplying one third of the
ratio with $\Lambda$, the total number of wedges in $\graph$. 
This
estimation is based on the fact that there are three different closed wedges
for each triangle
in $\graph$.

\subsubsection{Strengths of uniform wedge sampling}
This approach does not get adversely impacted from
the sparsity of graphs and can support small sampling probabilities as
each sampled entity (wedge) is an entity of interest and overall
search target (the number of triangles in the graph) is directly
proportional to the findings about sampled entities (ratio of closed wedges).   

\subsubsection{Deficiencies of uniform wedge sampling}
For graphs with power-law degree distributions, triangles are not
uniformly distributed across wedges, as shown on an example in Figure~\ref{fig:analysis}. In this figure, we present the percentage of wedges hinging at vertices with a given degree (indicated via green circles) and the percentage of triangles that vertices with a given degree participate in (indicated via red triangles) for a real-world power-law Web graph from Google~\cite{Kunegis:2013}. 
As can be seen from the figure,
 the uniform wedge sampling approach would
select more wedges hinging at higher degree vertices simply because
there are more such wedges. In fact, as expected, as the vertex degree
$d$ increases, the number of wedges hinging at vertices with degree
$d$ increases almost quadratically. However, the number of triangles
that vertices with degree $d$ participate in have a linear
relationship with the degree. 
These characteristics demonstrate that
uniform wedge sampling will provide biased estimations.
We observe similar characteristics in many real world networks and they
become more prominent as graphs grow larger and sparser.

\begin{figure}
    \centering
        \includegraphics[width=6.8cm]{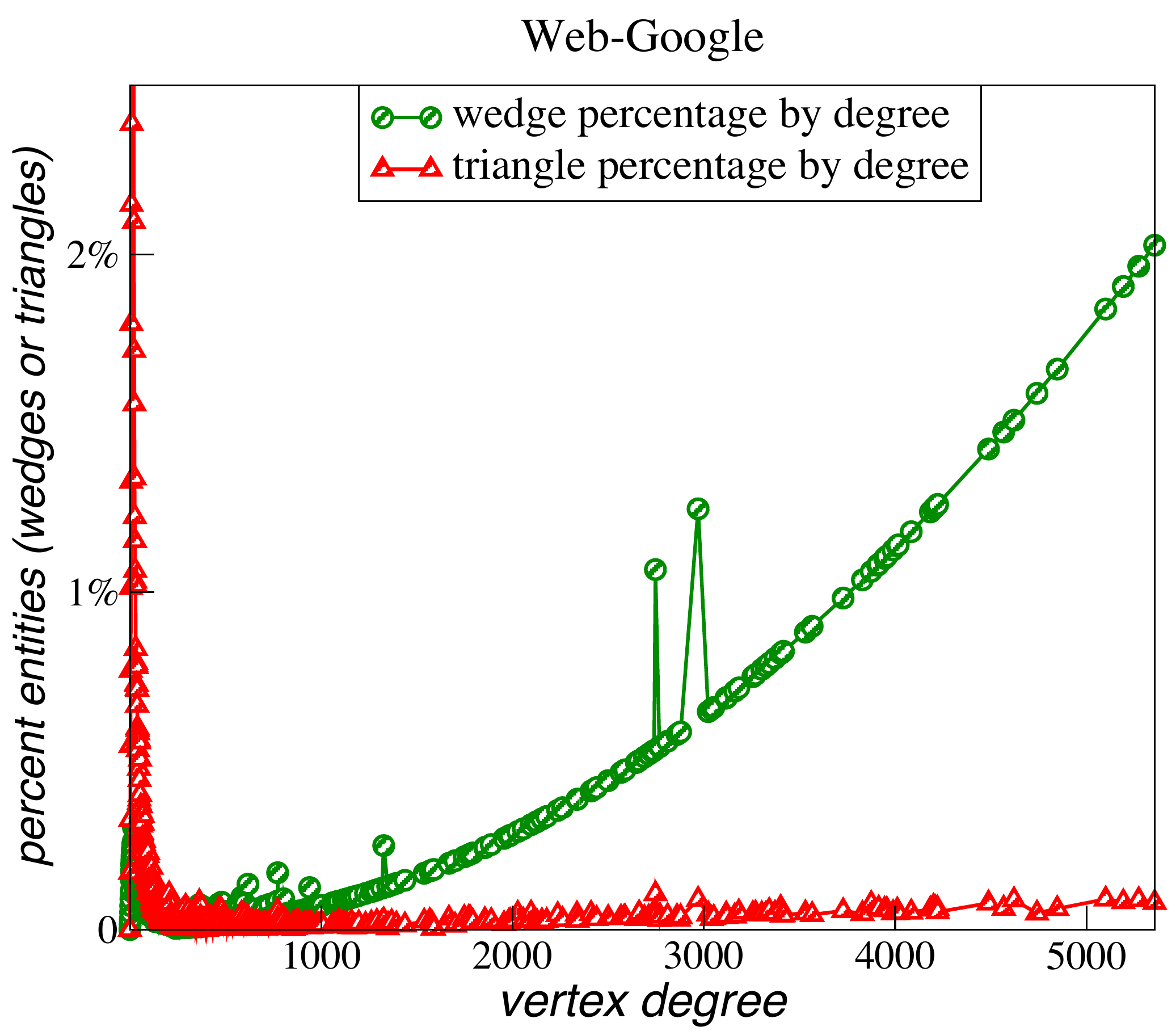}
        \caption{Wedge and triangle distribution per vertex degree for the Web-Google dataset~\cite{Kunegis:2013}.}
        \label{fig:analysis}
\vspace{-3.5ex}
\end{figure}

It is also important to note that, for applying uniform wedge
sampling, one needs to compute the degree distribution of the vertices
and the number of wedges in $\graph$; therefore, such approaches
require a preprocessing step which takes linear time in the number of edges.

\subsubsection{Proposed improvement over uniform wedge sampling}
%
To avoid favoring high degree vertices, our algorithm starts by random
edge sampling instead of uniform wedge sampling.  We then extend each
sampled edge to a wedge, and check for the existence of the closing
edge in $\graph$ to determine whether the vertices of the wedge define
a triangle or not. To increase the accuracy of our estimate, we only
consider the wedges that hinge at the lower degree vertex of the
selected edges. By selecting random edges in the first step, we ensure
that we start with a good sampling distribution. Note that the
probability of sampling a wedge in our schema is linearly
proportional with the degree of the vertex that the wedge hinges at,
whereas for uniform wedge sampling, that probability is
quadratically proportional with the degree of the hinge vertex.


\section{Edge-Based Wedge Sampling}
\label{sec:proposed_method}

\begin{figure*}
    \centering
    \begin{subfigure}[b]{0.23\textwidth}
        \includegraphics[width=\textwidth]{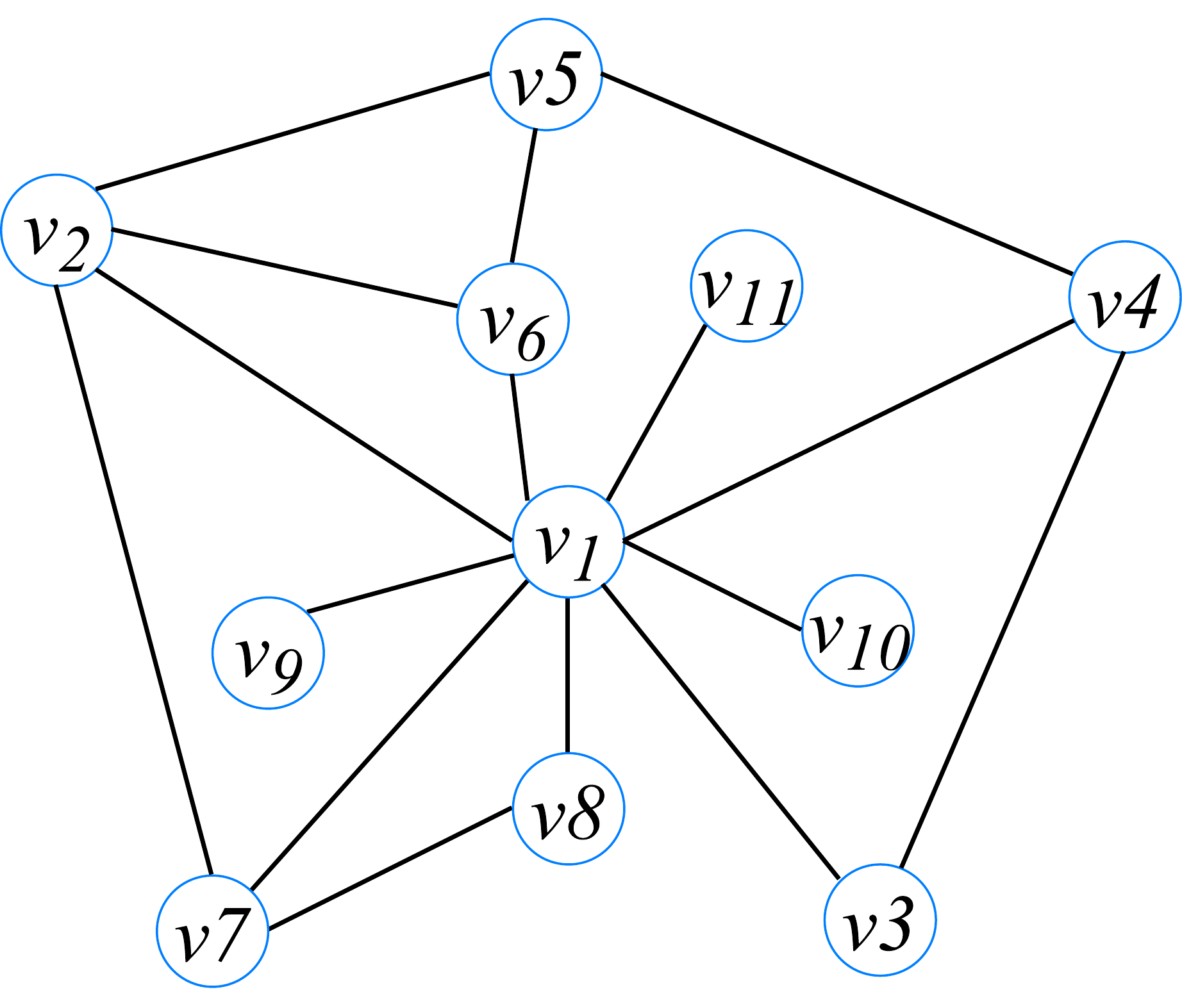}
        \caption{Example graph $\graph$ with 11 vertices, 16 edges, and 5 triangles. Sampling probability $p$ is given as $p$=$3/16$. }
        \label{fig:gull}
    \end{subfigure}
    ~ 
    \begin{subfigure}[b]{0.23\textwidth}
        \includegraphics[width=\textwidth]{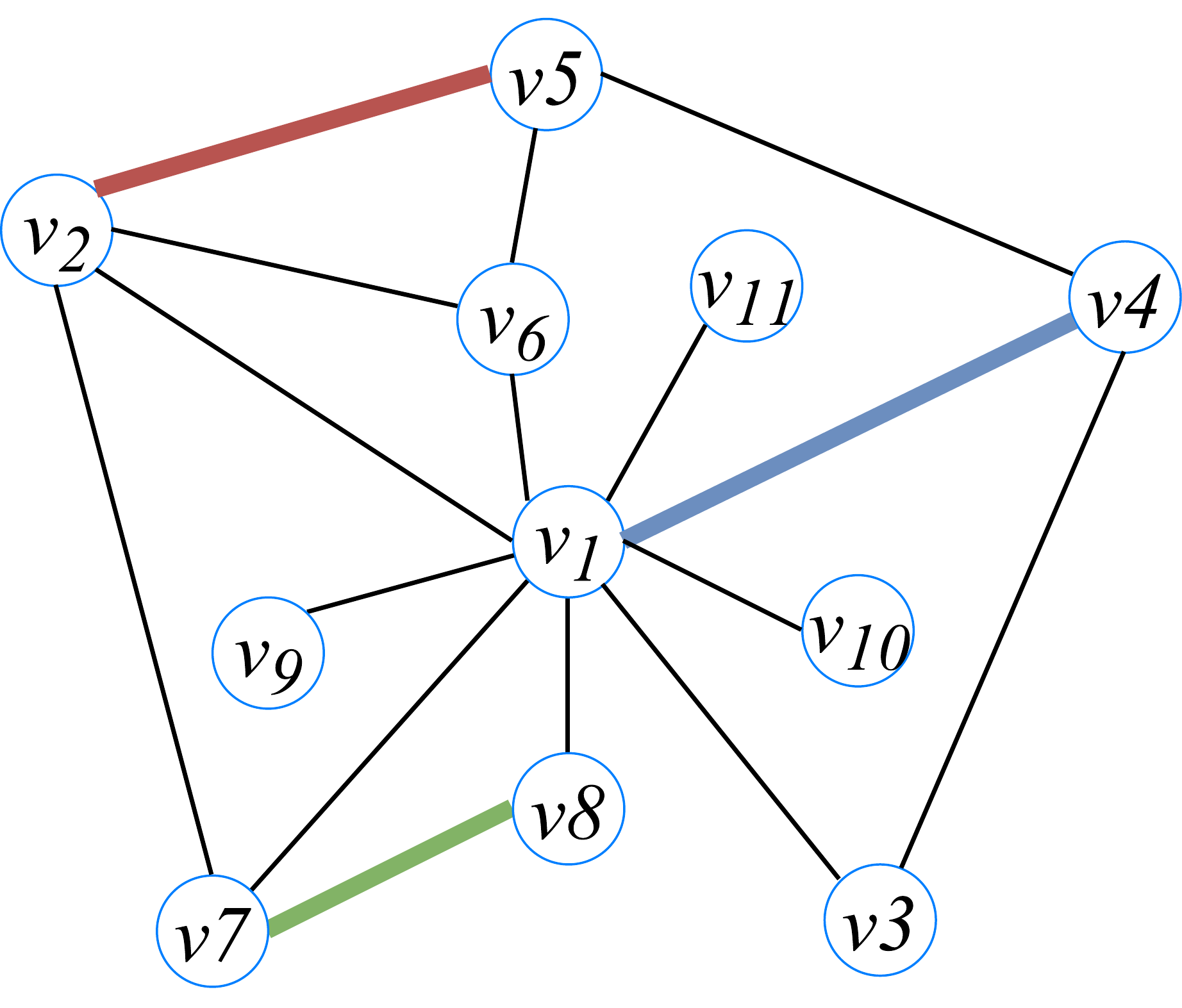}
        \caption{Random edge sampling over the example graph $\graph$ provides three selected edges: $\samp{\edges}$=${\{(v_2, v_5), (v_1, v_4), (v_7, v_8)\}}$.}
        \label{fig:tiger}
    \end{subfigure}
    ~ 
    \begin{subfigure}[b]{0.23\textwidth}
        \includegraphics[width=\textwidth]{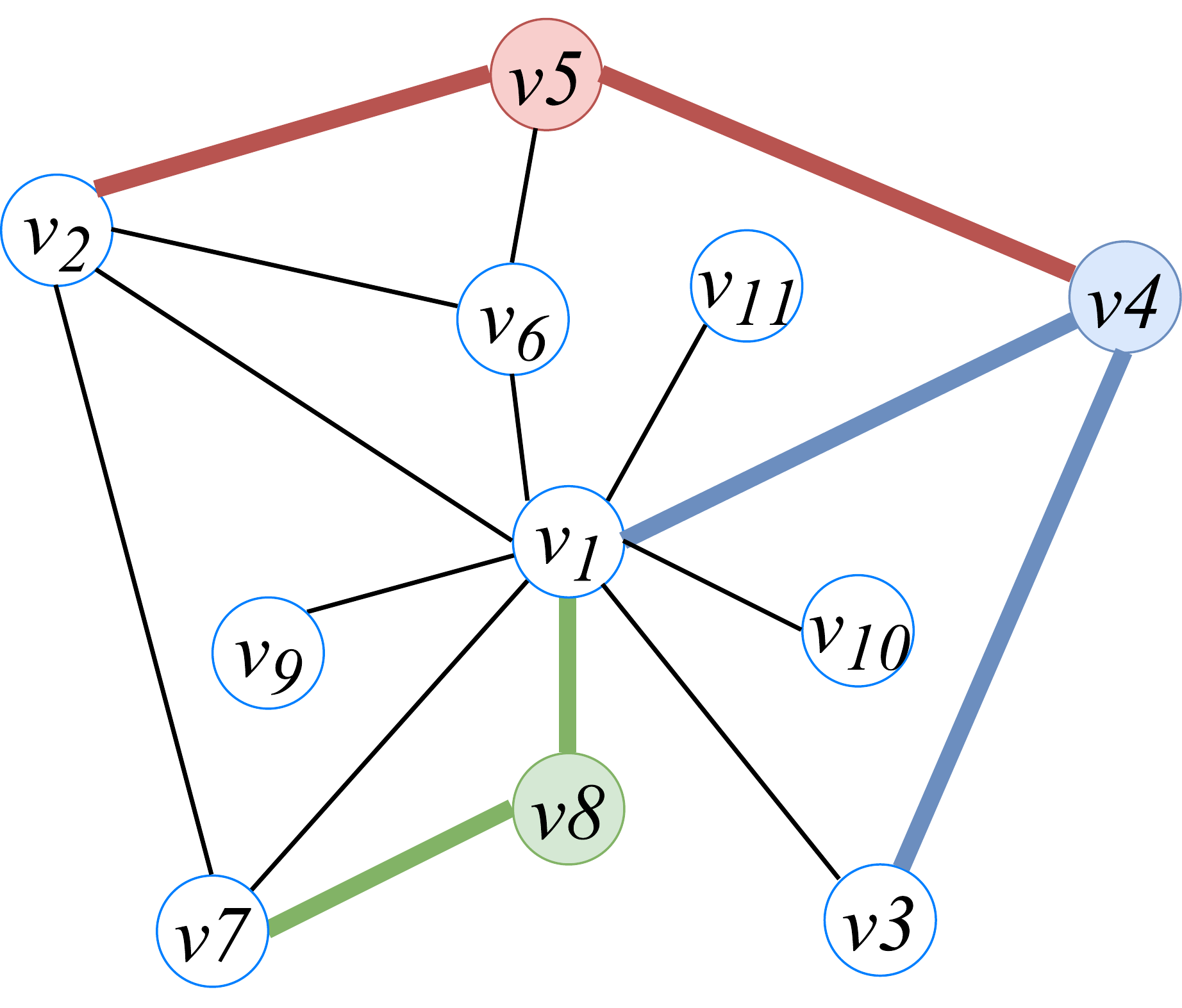}
        \caption{Informed wedge sampling for selected edges. Each wedge hinges on the lower degree end of the initially selected edges.}
        \label{fig:mouse}
    \end{subfigure}
    ~ 
    \begin{subfigure}[b]{0.23\textwidth}
        \includegraphics[width=\textwidth]{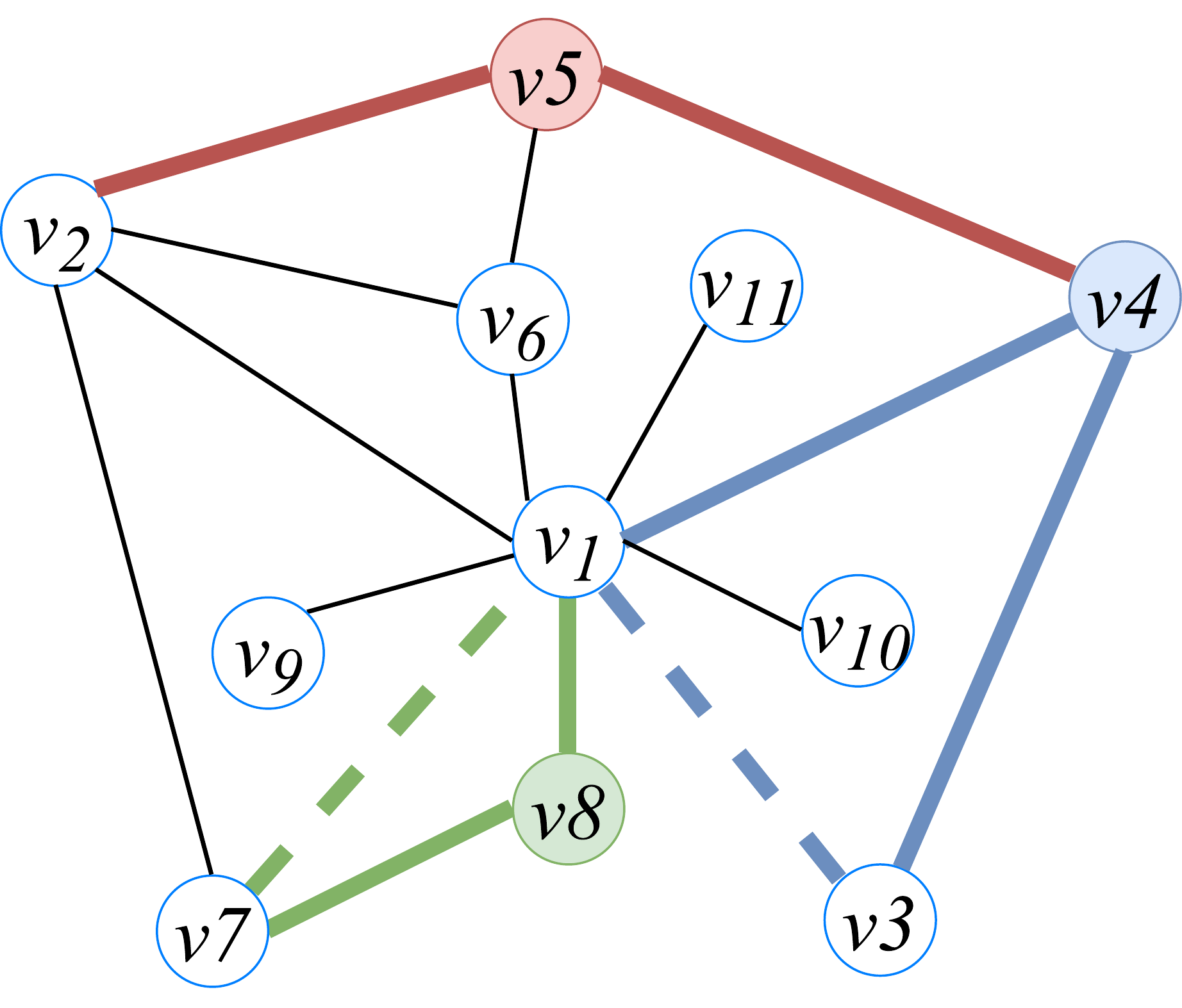}
        \caption{Triangle check in $\graph$. Wedges hinging at $v_4$, $v_8$ form triangles, $\tau$$=$$(d_{v_4}-1)$$+$$(d_{v_8}-1)$$=$$3$. Estimate is $\tau/3p\approx 5.33$. }
        \label{fig:mouse}
    \end{subfigure}
    \caption{Illustration of $\alg$ on an example graph.}
    \label{fig:illustration}
    \vspace{-1.5ex}
\end{figure*}

Our approach to estimating the
number of triangles in $\graph$ is based on sampling a set of edges $\samp{\edges}$
and estimating the number of triangles in $\graph$ that the sampled edges are
part of.
To accomplish this, after sampling a set of edges $\samp{\edges}$ in a
first step, we keep sampling further edges in a second step to form
wedges with each sampled edge in $\samp{\edges}$. 
%
%
To increase the accuracy of our estimate, in the second step, we only consider edges that form
wedges hinging at the lower degree vertices of the edges selected in the first step.  
In a third step, we perform
a check for the third edge to determine whether the vertices of a
selected wedge define a triangle or not.  
If any such wedge is closed, we estimate the number of triangles the
originally sampled edge contributes in to be the degree of its
lower-end vertex. If a sampled wedge is not closed, we estimate the
number of triangles the originally sampled edge contributes in to be
zero. Using these estimates, we then estimate the total number of triangles in the original graph.

Since our approach is initially based on random edge sampling, we do not need to perform a preprocessing step to obtain degree distributions of vertices in $\graph$; we only consider degrees of vertices of sampled edges $\samp{\edges}$. Furthermore, via edge sampling, we avoid being highly biased towards high degree vertices. By constructing wedges via a second sampling stage we ensure that all sampled edges turn into entities of interest. By selecting low-end of each sampled edge as wedge hinge points, we reduce the search space. Finally, by making an estimation associated with each sampled edge (i.e. the number of triangles that edge participates in), we ensure that our overall estimation correlates linearly with the sampling size, which enables us to significantly reduce sampling ratio.

Next, we discuss in detail our edge-based wedge selection and triangle estimation algorithm $\alg$. 
Pseudocode of $\alg$ is depicted in Algorithm~\ref{algo}.  Given a
graph $\graph$ and a sampling probability $p$, we first randomly sample edges
$\samp{\edges}$ as shown in lines 2--5. For any
edge $e=(u,v)\in \samp{\edges}$ selected in the first step,
let us assume $d_v \leq d_u$ without loss of generality. In lines
6--7, for any selected edge in $\samp{\edges}$, we randomly select one
more edge, say $(v,w)$ among the remaining edges incident to $v$ to
form a path or wedge $u$-$v$-$w$. In lines 8--9, we check for the
existence of the edge $(u,w)$ in the original graph.
If the original graph indeed contains the edge $(u,w)$, we increment
the total triangle count $\tau$ by $d_v - 1$.  That is, if the
sampled wedge originating from $e=(u,v)$ is closed and forms a
triangle, then we estimate $T(e)$, the number of triangles $e$
participates in, as the number of edges adjacent to $e$ at $v$, which is
equal to $d_v-1$.   The rationale of this estimation is based on a
single wedge, hence, we estimate $T(e)$ to be 0 or $d_v-1$ depending
on the outcome.
%
%
Finally, in line 10, we return $\tau / 3p$ as an estimate for $\Delta$
since a triangle can be counted by all of its three edges' estimates.
The complexity of our algorithm is $O(m)$: the subgraph sampling takes
$O(m)$ time, and the rest of the algorithm takes $O(pm)$ time.

We illustrate our algorithm $\alg$ on an example (see
Figure~\ref{fig:illustration}).  Our example graph $\graph$ contains
11 vertices, 16 edges and 5 triangles
(Figure~\ref{fig:illustration}(a)).  The triangles are $\{v_1, v_2,
v_6\}, \{v_1, v_2, v_7\}, \{v_1, v_3, v_4\}, \{v_1, v_7, v_8\},$ and
$\{v_2, v_5, v_6\}$.
Running $\alg$ with probability $p = 3/16$, let us assume that the
sampled edges $\samp{\edges}$ in the first step are $(v_2,v_5), (v_1,v_4),$ and $(v_7,v_8)$
(Figure~\ref{fig:illustration}(b)).
In the second step, we randomly select an edge incident to the low
degree vertices of these edges (Figure~\ref{fig:illustration}(c));
for $e = (v_2, v_5)$ let the randomly selected edge be $e'=(v_5,
v_4)$, for $e=(v_1,v_4)$ let $e'=(v_4,v_3)$, and for $e=(v_7,v_8)$ let
$e'=(v_8,v_1)$ (note that the last selection is with probability 1
since $d_{v_8} - 1 = 1$).
Then, we compute $\tau$ based on whether these wedges are closed or
not (Figure~\ref{fig:illustration}(d)): $v_2$$-$$v_5$$-$$v_4$ is not closed so that wedge does not
increase $\tau$, but the other two wedges $v_1$$-$$v_4$$-$$v_3$ and
$v_7$$-$$v_8$$-$$v_1$ are closed and they both increase $\tau$ by
their low degrees  - 1 respectively.  In particular,  $\tau$ gets increased
by 2 for the edge $(v_1,v_4)$ and by 1 for the edge $(v_7, v_8)$.  
Finally, $\alg$ outputs the final estimate as:
\[
\frac{\tau}{3p} = \frac{3}{3 \times \frac{3}{16}} = \frac{16}{3} \approx 5.33
\]

To compare our approach to existing approaches,
Figure~\ref{fig:others}(a) showcases a potential run of the algorithm
proposed in~\cite{Seshadhri:2014} with a sampling rate of $p = 3/16$, i.e.,
$k=3$ random wedges.  As seen in the figure, their randomly selected
wedges are $v_1$$-$$v_4$$-$$v_5$, $v_1$$-$$v_2$$-$$v_7$, and
$v_{11}$$-$$v_1$$-$$v_3$, and only one of them is closed:
$v_1$$-$$v_2$$-$$v_7$.
In this graph number of wedges $\Lambda = 56$, therefore their estimate can
be calculated as:
\[
1 \times\frac{\Lambda}{3k} = \frac{56}{9} \approx 6.22
\]

Figure~\ref{fig:others}(b) showcases a potential run of the algorithm
proposed in~\cite{Etemadi:2016}.
One can argue that since our approach and uniform wedge sampling work
with wedges, the random edge sampling method should be allowed to pick
twice the number of edges. However, note that random edge selection
requires constructing the subgraph $\samp{\graph}$, and identifying
all the wedges in $\samp{\graph}$, which can be very costly.
Practically, sampling an edge and sampling a wedge have similar costs
and algorithms should be compared based on their performance when they
sample the same number of entities (edges or wedges).

We note that when $p$ is very low, finding even one closed wedge using the algorithm in~\cite{Etemadi:2016} 
becomes very difficult even in our simple example. So we pick twice the number of
edges to be able to find some closed wedges in the sampled graph,
i.e., we use a sampling rate of $2p$. 
As seen in the figure, their sample graph contains 6 edges and these sampled edges
form 4 wedges. Out of these 4 wedges, $\samp{\Lambda}^+=3$ of them are
closed: $v_5$$-$$v_2$$-$$v_6$, $v_1$$-$$v_2$$-$$v_6$, and
$v_1$$-$$v_4$$-$$v_3$.
Noting that the probability used in their first step is $3/8$, their
estimate can be computed as:
\[
\frac{\samp{\Lambda}^+}{3p^2} = \frac{3}{3 \times (\frac{3}{8})^2} =
\frac{64}{9} \approx 7.11
\]

\begin{algorithm}[t]
\caption{$\alg$($\graph = (\vertices, \edges), p$)
}
\label{algo}
\begin{algorithmic}[1]
%
\STATE $\tau \gets 0$, $\ \samp{\edges} \gets \emptyset$
\FOR{{\bf each} edge $e \in \edges$}
    \STATE randomly select edge $e=(u,v)$ with probability $p$
    \IF{$e$ is selected}
        \STATE{$\samp{\edges} \leftarrow \samp{\edges} \cup {\{e\}}$}
    \ENDIF
\ENDFOR

\FOR{{\bf each} selected edge $e = (u,v) \in \samp{\edges}$}
    \STATE randomly select an edge $e'=(v,w)$ incident to $v$, where
$d_v \leq d_u$ and $w \neq u$
\IF{$e''=(u,w) \in \edges$}
\STATE{$\tau \gets \tau + d_v - 1$}
\ENDIF
\ENDFOR
\STATE return $\tau / 3p$
\end{algorithmic}
\end{algorithm}

\begin{figure}
\centering
    \begin{subfigure}[b]{0.22\textwidth}
        \includegraphics[width=\textwidth]{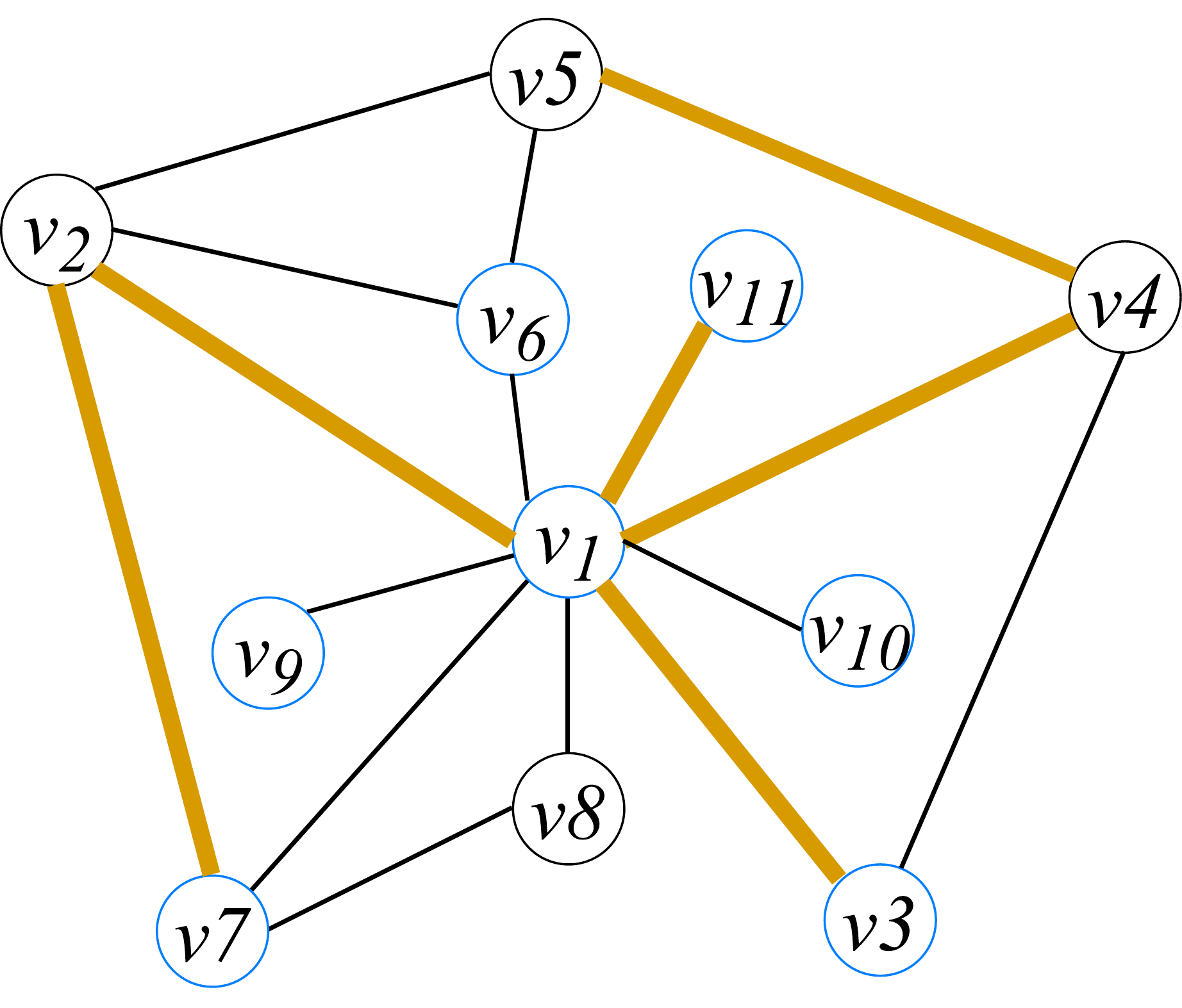}
        \caption{Random-wedge sampling algorithm proposed
in~\cite{Seshadhri:2014}. Out of the 3 sampled wedges, only 1 wedge is
closed. Estimate is $56/9$$\approx$$6.22$.} \label{fig:mouse}        
    \end{subfigure}
    ~
    \begin{subfigure}[b]{0.22\textwidth}
        \includegraphics[width=\textwidth]{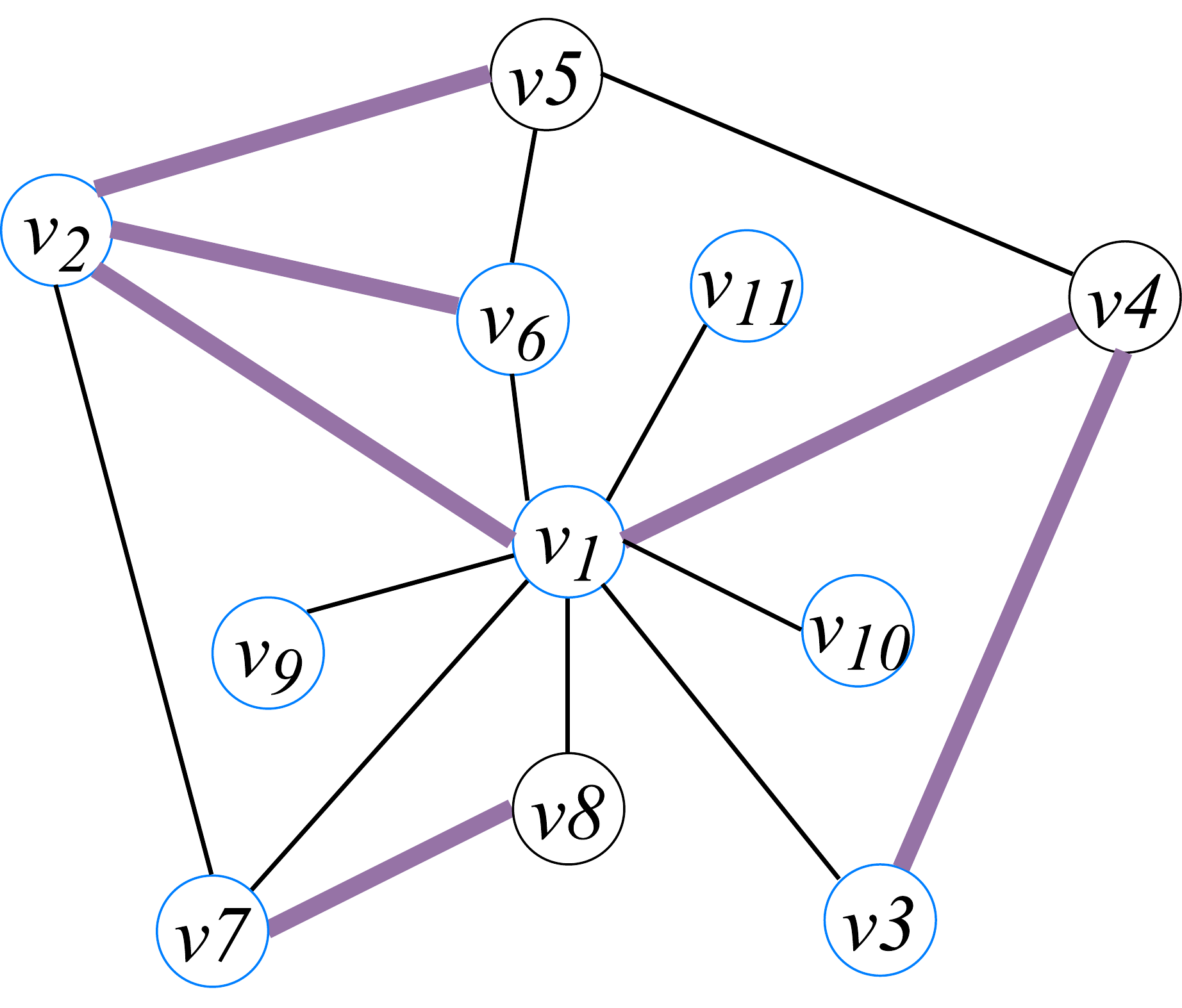}
        \caption{Random-edge sampling and wedge counting algorithm
        proposed in~\cite{Etemadi:2016}. The edges form 4 wedges, 3 of which are
        closed. Estimate is $64/9$$\approx$$7.11$.} \label{fig:mouse}
    \end{subfigure}
    \caption{Illustration of applying the state-of-the-art approaches:
(a) uniform wedge sampling approach proposed
in~\cite{Seshadhri:2014}, and
 (b) random edge sampling approach proposed
in~\cite{Etemadi:2016}.}
    \label{fig:others}
\vspace{-1.5ex}
\end{figure}


\section{Expected Value, Variance, and RSE}
\label{sec:rse}

We analyze the expected value, variance, and relative
standard error (RSE) of our estimate $\tau$ in order to investigate
the confidence interval of the outputs of $\alg$. We also analyze the
same terms for the state-of-the-art approaches to provide theoretical comparisons. 
We provide the theoretical analysis in this section, and show empirical evidences that corroborate our theory in the next section.

\subsection{Expected Value, Variance, and RSE of $\alg$}
To provide this analysis we first express our output in terms of indicator variables.
First, we number the edges in the original graph as $e_1, \ldots,
e_m$.  For any given edge $e_i$, we represent each of the $T(e_i)$
triangles that $e_i$ participates in by the wedges that hinge on the
lower degree vertex of $e_i$.  To that extent, we number each such
edge that belongs to a triangle with $e_i$ and adjacent to $e_i$ on
the lower degree vertex.  Letting these edges be $e_{i1}, \ldots,
e_{iT(e_i)}$, we define our indicator variable $w_{ij}$ as the
indicator for selecting $e_i$ in the first step and $e_{ij}$ in the
second step as the wedge for $e_i$.
Finally, defining $\delta(e_i)$ to be the degree of the lower degree
vertex of $e_i$, we express $\tau$, our estimate for $\samp{T}$ as:
\[
\tau = \sum_{i=1}^m \sum_{j=1}^{T(e_i)} w_{ij} \cdot
(\lowdeg{e_i} - 1).
\]

\noindent Using this expression, we show that the expected value of $\tau$ is
$3p\Delta$, and thus the expected value of the output of $\alg$ is
$\Delta$.

\begin{lemma}
The expected value of $\tau$ is $3p\Delta$.
\label{lem:expected}
\end{lemma}
\begin{proof}
We have $\tau = \sum_{i=1}^m \sum_{j=1}^{T(e_i)} w_{ij}
\cdot (\lowdeg{e_i} - 1)$ for the $\alg$ algorithm.  Hence:
\[%
\mathbb{E}(\tau) = \sum_{i=1}^m \sum_{j=1}^{T(e_i)}
\mathbb{E}(w_{ij}) \cdot (\lowdeg{e_i} - 1)
\]
The probability that $w_{ij} = 1$ is $p/ (\lowdeg{e_i} -
1)$, therefore:
\[
\mathbb{E}(\tau) = \sum_{i=1}^m
\sum_{j=1}^{T(e_i)} p = \sum_{i=1}^m p T(e_i) = 3p\Delta
\]
where the latter equality holds as each triangle is counted in each of
its edges' $T(e)$ values.
\end{proof}

Next, we argue that the variance of $\tau$ is low, however, our
arguments involve some features of the original graph $\graph$.  More
specifically, the variance of $\tau$ involves the terms $K$ and
$\phi$, where $K$ is the number of pairs of triangles that share an
edge in $\graph$~\cite{Etemadi:2016}, and we define $\phi$ as $\phi =
\sum_{i=1}^\Delta \sigma(t_i) - 3$, where $\sigma(t_i)$ is the sum of
$\delta(e)$ for edges in triangle $t_i$: $\sigma(t_i) = \sum_{e \in
t_i} \delta(e)$. Intuitively, one can think of $\phi$ as the sum of
the degrees of each vertex in each triangle in $\graph$, except that
instead of summing the maximum degree vertex, the sum involves the
minimum degree vertex twice.  This is due to fact that we select the low degree
vertex of the initially selected edge as hinge point in our 
edge-based wedge construction mechanism.  

\begin{lemma}
The variance of $\tau$ is $p \phi - p^2 (3\Delta + 2K)$, where $\phi = 
 \sum_{i=1}^\Delta \sigma(t_i) - 3$.
\label{lem:variance}
\end{lemma}
\begin{proof}
Since the indicator variables $w_{i_1j_1}$ and
$w_{i_2j_2}$ are independent for $i_1 \neq i_2$, we
have:
\[
var(\tau) = \sum_{i=1}^m var \left(\sum_{j=1}^{T(e_i)}
w_{ij} \cdot (\lowdeg{e_i} - 1) \right)
\]
This can be rewritten as
\begin{align*}
var(\tau) &= \sum_{i=1}^m \left( \sum_{j=1}^{T(e_i)} (\lowdeg{e_i}
- 1)^2 \cdot var(w_{ij}) \right. \\
&+ \left. \sum_{j\neq k} (\lowdeg{e_i} - 1)^2 \cdot covar(w_{ij},
w_{ik}) \right)
\end{align*}
We have \mbox{$var(w_{ij}) = (p/(\lowdeg{e_i} - 1) -
p^2/(\lowdeg{e_i}-1)^2)$} and since we are selecting only one wedge
for each 
edge, we have  $covar(w_{ij}, w_{ik}) = (0 - p^2/(\lowdeg{e_i}-1)^2)$.
Further simplifying:
\begin{align*}
var(\tau) &= \sum_{i=1}^m \left( \sum_{j=1}^{T(e_i)}
 p \cdot (\lowdeg{e_i} - 1) - \sum_{j=1}^{T(e_i)} p^2  -  \sum_{j\neq k} p^2 \right)\\
&= p\sum_{i=1}^m T(e_i) \cdot (\lowdeg{e_i} - 1) - p^2 \sum_{i=1}^m
T(e_i)\\
& - p^2 \sum_{i=1}^m T(e_i)(T(e_i) - 1)
\end{align*}
The first term can be written as a sum over the triangles where for
each triangle, the value $\lowdeg{e} - 1$ is summed for each of its
edges.  The second term can be written in terms of $\Delta$, the
number of triangles in $\graph$.  And, the third term can be written
in terms of $K$, the number of pairs of triangles that share an edge
in $\graph$.  Finally, we prove:
\[
var(\tau) = p\left( \sum_{i=1}^\Delta \sigma(t_i) - 3\right) - p^2 (3\Delta + 2K)
\]
where $\sigma(t_i)$ is the sum of $\delta(e)$ for edges in triangle
$t_i$.
\end{proof}

In order to compare our method to previous work, we use the Relative Standard Error (RSE) values of all of the methods and base our arguments on the RSE values.
Applying the formula for the relative standard error $RSE(\tau) =
\sqrt{var(\tau)}/\mathbb{E}(\tau)$, we obtain:
\begin{align*}
RSE(\tau) &=  
\frac{\sqrt{p\phi - p^2
(3\Delta + 2K)}}{3p\Delta}=
\sqrt{\frac{\phi}{9p\Delta^2} -
\frac{3\Delta + 2K}{9\Delta^2}}
\end{align*}

The above formula is complex, however, we can approximate it by simply
dropping the negative term inside the square root.  We argue that this
will not affect our actual RSE values too much since the $p$ values
we choose for our experiments will be very low, and in addition the
actual value of the negative term, $(3\Delta + 2K)/9\Delta^2$, will be
too low for the graphs we consider making the negative term very
insignificant.  Therefore, we approximate our RSE using the
following formula:
\begin{equation}
RSE(\tau) \approx \sqrt{\frac{\phi}{9p\Delta^2}}
\label{eqn:rse-tau}
\end{equation}
where $\phi = \sum_{i=1}^\Delta \sigma(t_i) - 3$ is a constant
depending only on the original graph $\graph$.  In our experiments,
we show that we achieve RSE values that are very close to the
approximation we provide in Equation~\ref{eqn:rse-tau}.

\subsection{RSE of Other Methods}

We also analyze the expected value, variance, and the RSE of the other
methods to offer a theoretical comparison with our method.
We consider two state-of-the-art methods, the method introduced by Etemadi et al.~\cite{Etemadi:2016}, which we call $\alges$, as well as the method introduced by Seshadhri at al.~\cite{Seshadhri:2014}, which we call $\algws$.

We begin with the analysis of $\algws$, the method introduced by
Seshadhri et al.~\cite{Seshadhri:2014}.  We define indicator variables
for each of the $3\Delta$ closed wedges in order to express the random
variable that counts the number of closed wedges observed in $\algws$.
Letting $w_i$ to be an indicator for the $i^{th}$ closed wedge,
indicator for whether it is one of the randomly selected $k = pm$
wedges, we define:

\[
\omega = \sum_{i=1}^{3\Delta} w_i
\]

Using $\omega$, $\algws$ outputs $\omega \Lambda / 3k$ to estimate
$\Delta$.  In the next lemma, we provide the expected value and the
variance analysis of $\omega$ to be used in the relative standard
error analysis of $\omega$, where $C = 3\Delta / \Lambda$ is the
clustering coefficient:

\begin{lemma}
The expected value of $\omega$ is $p mC$, and the variance of $\omega$
is $pm C(1-C) \left(1 - \frac{p m - 1}{\Lambda - 1}\right)$.
\end{lemma}


\begin{proof}
We have \mbox{$\mathbb{E}(\omega) = \sum_{i=1}^{3\Delta}
\mathbb{E}(w_i) = \sum_{i=1}^{3\Delta} k / \Lambda = 3\Delta k /
\Lambda$}, which implies $\mathbb{E}(\omega) = p m C$ by definition.

For variance, we have $var(w_i) = \frac{k}{\Lambda} (1 -
\frac{k}{\Lambda})$.  Also, 
we have
$covar(w_j,w_k) = 
 \frac{k}{\Lambda} (\frac{k-1}{\Lambda -
1} -\frac{k}{\Lambda})$
for any $w_j \neq w_k$. Therefore, 
 we write variance as:
\begin{align*}
var(\omega) &= \sum_{i=1}^{3\Delta} var(r_i) + \sum_{j \neq k} covar(r_j, r_k) \\
&= 3\Delta \frac{k}{\Lambda} \left(1 - \frac{k}{\Lambda}\right) 
+ 3\Delta ( 3\Delta - 1) \frac{k}{\Lambda} \left(\frac{k-1}{\Lambda -
1} -\frac{k}{\Lambda}\right)\\
&= kC \left(1 - \frac{k}{\Lambda}\right) 
+ kC ( 3\Delta - 1) \left(\frac{k-1}{\Lambda - 1} -\frac{k}{\Lambda}\right)\\
&= kC \left(1 - \frac{k - 1}{\Lambda - 1}
+  3\Delta \left(\frac{k-1}{\Lambda - 1} - \frac{k-1}{\Lambda} -
\frac{1}{\Lambda}\right) \right) \\
&= kC \left(1 - \frac{k - 1}{\Lambda - 1}
+  3\Delta \frac{k - 1}{\Lambda(\Lambda - 1)} - 3\Delta
\frac{1}{\Lambda}\right)
\end{align*}

The last equation simplifies as $pmC (1 - C) \left(1 - \frac{p m -
1}{\Lambda - 1}\right)$ proving the variance stated in the lemma.
\end{proof}

Using this lemma, we obtain the formula for their relative
standard error $RSE(\omega) = \sqrt{var(\omega)}/\mathbb{E}(\omega)$:

\begin{align*}
RSE(\omega) &=  \frac
{\sqrt{pmC (1 - C) \left(1 - \frac{p m - 1}{\Lambda - 1}\right)}}{pmC}\\
 &=  \sqrt{\frac{1-C}{pmC} \left(1 - \frac{p m - 1}{\Lambda - 1}\right)}
\end{align*}

Similarly, we can drop the insignificant negative terms from the above
formula and get an approximation with fewer simpler terms:
\begin{equation}
RSE(\omega) \approx \sqrt{\frac{1-C}{p m C}}
\label{eqn:rse-omega}
\end{equation}

\begin{table*}
\caption{Datasets used in experiments and their features.}
\label{tbl:datasets}
\centering
\begin{tabular}{l r r r r r r r l}
Dataset & $n$ & $m$ & $\Delta$ &  $C$$=$$\frac{3\Delta}{\Lambda}$ &
$\frac{3\Delta}{m}$ & $\frac{\phi}{3\Delta}$ & $\frac{K}{\Delta}$ & Description 
\vspace{1mm}
\\
\hline  
Ego-Facebook~\cite{Leskovec:2015} &      4K &      88K &     1612K  &   0.5192 &       54.8 &      115.5 &   141.9 & Online social network (OSN) in Facebook \\
Enron-email~\cite{Kunegis:2013}   &     36K &     183K &      727K  &   0.0853 &       11.9 &      102.1 &    50.2 & Email communication network in Enron \\
Brightkite~\cite{Leskovec:2015}   &     58K &     214K &      494K  &   0.1106 &        6.9 &       76.6 &    59.1 & OSN in Brightkite \\
Dblp-coauthor~\cite{Leskovec:2015}&    317K &    1049K &     2224K  &   0.3064 &        6.4 &       37.6 &    47.2 & Co-authorship network in DBLP\\
Amazon~\cite{Leskovec:2015}       &    334K &     925K &      667K  &   0.2052 &        2.2 &        7.1 &     5.3 & Co-purchasing network from Amazon \\
Web-NotreDame~\cite{Kunegis:2013} &    325K &    1090K &     8910K  &   0.0877 &       24.5 &      119.4 &   174.2 & Web graph of Notre Dame \\
Citeseer~\cite{Kunegis:2013}      &    384K &    1736K &     1351K  &   0.0496 &        2.3 &       26.6 &    11.6 & Citation network in Citeseer \\
Dogster~\cite{Kunegis:2013}       &    426K &    8543K &    83499K  &   0.0143 &       29.3 &     1085.2 &   503.8 & OSN from dogster.com website \\
Web-Google~\cite{Kunegis:2013}    &    875K &    4322K &    13391K  &   0.0552 &        9.3 &       35.4 &    46.4 & Web graph from Google \\
YouTube~\cite{Leskovec:2015}      &   1134K &    2987K &     3056K  &   0.0062 &        3.1 &      257.7 &    82.4 & OSN in Youtube \\
DBLP~\cite{Kunegis:2013}          &   1314K &    5362K &    12184K  &   0.1703 &        6.8 &       38.0 &    35.8 & Co-authorship network in DBLP \\
As-skitter~\cite{Leskovec:2015}   &   1696K &   11095K &    28769K  &   0.0054 &        7.8 &      796.4 &   713.3 & Internet connections from Skitter project \\
Flicker~\cite{Kunegis:2013}       &   2302K &   22838K &   837605K  &   0.1076 &      110.0 &     1239.9 &   732.8 & Online social network in Flicker \\
Orkut~\cite{Kunegis:2013}         &   3072K &  117185K &   627584K  &   0.0413 &       16.1 &      371.3 &   106.9 & Online social network in Orkut \\
LiveJournal~\cite{Leskovec:2015}  &   3997K &   34681K &   177820K  &   0.1253 &       15.4 &      204.3 &   222.1 & OSN in LiveJournal \\
Orkut2~\cite{Boldi:2011}          &  11514K &  327036K &   223127K  &   0.0003 &        2.0 &     1229.0 &   155.4 & OSN in Orkut \\
Web-Arabic~\cite{Boldi:2011}      &  22743K &  553903K & 36895360K  &   0.0313 &      199.8 &     2047.7 &  3042.7 & Web graph from Arabian countries \\
MicrosoftAG~\cite{MicrosoftAG:2015}&  46742K &  528463K &   578188K &   0.0151 &        3.3 &       97.1 &    33.9 & Citation network from Microsoft Academic \\
Twitter~\cite{Kunegis:2013}        &  41652K & 1202513K & 34824916K &   0.0008 &       86.9 &    11638.5 &  5061.5 & OSN from Twitter \\
Friendster~\cite{Kunegis:2013}     &  65608K & 1806067K &  4173724K &   0.0174 &        6.9 &      311.6 &    44.4 & OSN of website Friendster \\
\hline
\end{tabular}
\end{table*}

In the next section, we provide empirical evidence that verifies the
above approximation given in Equation~\ref{eqn:rse-omega}.

Next, we analyze $\alges$, the method proposed by Etemadi et
al.~\cite{Etemadi:2016}.  Once again, we define an indicator variable
for each of the $3\Delta$ closed wedges in order to express the closed
wedge count of this method.  Letting $r_i$ to be an indicator for the
$i^{th}$ closed wedge, indicator for whether it is included in
$\samp{\graph}$ in the first step of $\alges$, we define:

\[
\rho = \sum_{i=1}^{3\Delta} r_i
\]

Using $\rho$, $\alges$ outputs $\rho / 3p^2$ to estimate $\Delta$.  In
the next lemma, we provide the expected value and variance of $\rho$
based on the discussion from their original work~\cite{Etemadi:2016}:

\begin{lemma}
The expected value of $\rho$ is $3p^2\Delta$, and the variance of
$\rho$ is $3\Delta(p^2-p^4)+8K(p^3-p^4)$.
\end{lemma}

Using this lemma, we obtain the formula for their relative
standard error $RSE(\rho) = \sqrt{var(\rho)}/\mathbb{E}(\rho)$:

\begin{align*}
RSE(\rho) &=  \frac{\sqrt{3\Delta(p^2-p^4)+8K(p^3-p^4)}}{3p^2\Delta} \\
 &=  \sqrt{\frac{1}{3p^2\Delta} - \frac{1}{3\Delta} +
           \frac{8K}{9p\Delta^2} - \frac{8K}{9\Delta^2}} 
\end{align*}

Once again, we can drop the insignificant negative terms from the
above formula to get an approximation with fewer simpler terms.
\begin{equation}
RSE(\rho) \approx \sqrt{\frac{1}{3p^2\Delta} + \frac{8K}{9p\Delta^2}}
\label{eqn:rse-rho}
\end{equation}

Note that the approximation in Equation~\ref{eqn:rse-rho} is different than
the one in proposed in the work of Etemadi et al.~\cite{Etemadi:2016}
and the main reason to include the extra term is to ensure a much
better approximation.  In the next section, we verify the validity of
our approximation schemes and argue that we can simply use the
theoretical approximations to compare the performance of the three methods.


\section{Evaluation}
\label{sec:eval}

In this section, we provide an empirical and theoretical evaluation of the proposed algorithm.

\subsection{Experimental Setup}
We compare the performance of our edge-based wedge sampling algorithm
(EWS) with the state-of-the-art 
edge sampling (ES) and wedge sampling (WS) approaches proposed by Etemadi et al.~\cite{Etemadi:2016} and Seshadhri et
al.~\cite{Seshadhri:2014} both theoretically and experimentally using graphs modeling real-world large-scale networks. 
We also corroborate the theoretical analysis we provided in Section~\ref{sec:rse} with empirical evidence. In the forthcoming experiments and analysis, we sample the same number of entities for the three different approaches; edges for ES, and wedges for WS and EWS. All approaches perform closed wedge checks as necessary; however, we do not account for the cost of those checks.  In practice, the costs are almost exactly the same for WS and EWS, however, ES performs more checks as it considers closed wedge checks for any pair of adjacent edges in the sampled subgraph.

For comparison and also in order to verify our theoretical analysis, we observe the relative standard error ($RSE$) values of different algorithms for each of the datasets and compare these values with the theoretical analysis we provided in Section~\ref{sec:rse}. 

In order to provide comparable results, as was done in~\cite{Etemadi:2016}, we either fix the sampling probability, and report the corresponding RSE values of the three algorithms, or fix an RSE value and report the corresponding sampling rates. 
We calculate the experimental $RSE$ values observed with $k$ different runs of the algorithms using the formula:
\[
RSE = \frac{\sqrt{\frac{1}{k}\sum_{i=1}^k(\Delta_i - \mu)}}{\Delta},
\]
where $\Delta_i$ is the estimate obtained in the $i$-th run and $\mu$
is the mean of all of the $k$ runs, i.e., $\mu =
\frac{1}{k}\sum_{i=1}^k \Delta_i$.  In all of our experiments, we
use $k=1000$.


The experiments are conducted on a server with 16 CPUs and 64 GB of RAM.
In our experimental analysis we use 20 real-world datasets that were also used in~\cite{Etemadi:2016}. These datasets vary in size; largest containing vertices in the order of ten millions and triangles in the order of ten billions, and the smallest
containing vertices in the order of thousands and triangles in the order of millions. Features of these datasets are presented in Table~\ref{tbl:datasets}. 

In addition to graph features, we also provide metrics of graphs that impact the performance of ES, WS and EWS in Table~\ref{tbl:datasets} as well. For example the higher the value of global clustering coefficient $C=3\Delta/\Lambda$, the better the performance of WS, as the chances of finding closed wedges during sampling increase.
Similarly, $3\Delta/m$ indicates the number of triangles an average edge participates in a graph, and the higher this number the better the performance of EWS in general. EWS performance also gets impacted from the value of $\phi/3\Delta$ as well, which can be considered as the degree of the lowest degree or medium degree vertex in an average triangle. The higher this number gets, the courser the estimations made by EWS, hence the worse its performance. A similar argument can be made for ES and $K/\Delta$.  

%


\subsection{Empirical Analysis and Verification of Theory}

Figure~\ref{fig:rse}, provides an empirical analysis of ES, WS, and
EWS over the 20 datasets listed in Table~\ref{tbl:datasets}. The
figure also depicts the RSE estimations indicated by Equations~\ref{eqn:rse-tau},~\ref{eqn:rse-omega}, and~\ref{eqn:rse-rho} for EWS, WS, and ES, respectively. In this figure, the $x$-axis ranges between the sampling probabilities that provides $RSE$ values between 0.50 and 0.04 for EWS and we report the corresponding RSE values for ES, WS, and EWS when using the same sampling rates.  

For all datasets depicted in Figure~\ref{fig:rse}, the empirical RSE
observations we make for EWS, WS, and ES match the theoretical
bounds EWS-est, WS-est, and ES-est that we compute using
Equations~\ref{eqn:rse-tau},~\ref{eqn:rse-omega}, and~\ref{eqn:rse-rho}, respectively. This proves the validity of our theoretical analysis.

Interestingly, we can observe in Figure~\ref{fig:rse} that the
datapoints for EWS and WS follow a similar pattern and they exhibit
parallel trends irrespective of the sampling ratio. 
In fact, this is to be expected as the theoretical $RSE$ approximations
in Equations~\ref{eqn:rse-tau} and \ref{eqn:rse-omega} both depend
linearly on the sampling ratio and hence the estimations EWS-est and
WS-est are parallel.  Therefore, the ratio of their RSE values
remain the same accross various sampling probabilities.
Whichever of these two algorithms will perform better is
based on the parameters of the graph.
In 14 of the datasets EWS performs better than WS and in
6 datasets WS performs better than EWS. We note that since this is a
log-log figure, the slight difference between the two lines in fact
indicate a multiplicative difference.

In Figure~\ref{fig:rse}, for all datasets and most $p$ values, EWS and
WS perform significantly better than ES. On the other hand, the slope
for ES-est is much steeper---also to be expected because of the
quadratic dependence on the sampling ratio---and even though ES provides much higher
RSE values for low $p$ values, as $p$ increases, the accuracy of ES
increases significantly.

\begin{figure*}
    \centering
    \begin{subfigure}[b]{0.245\textwidth}
        \includegraphics[width=\textwidth]{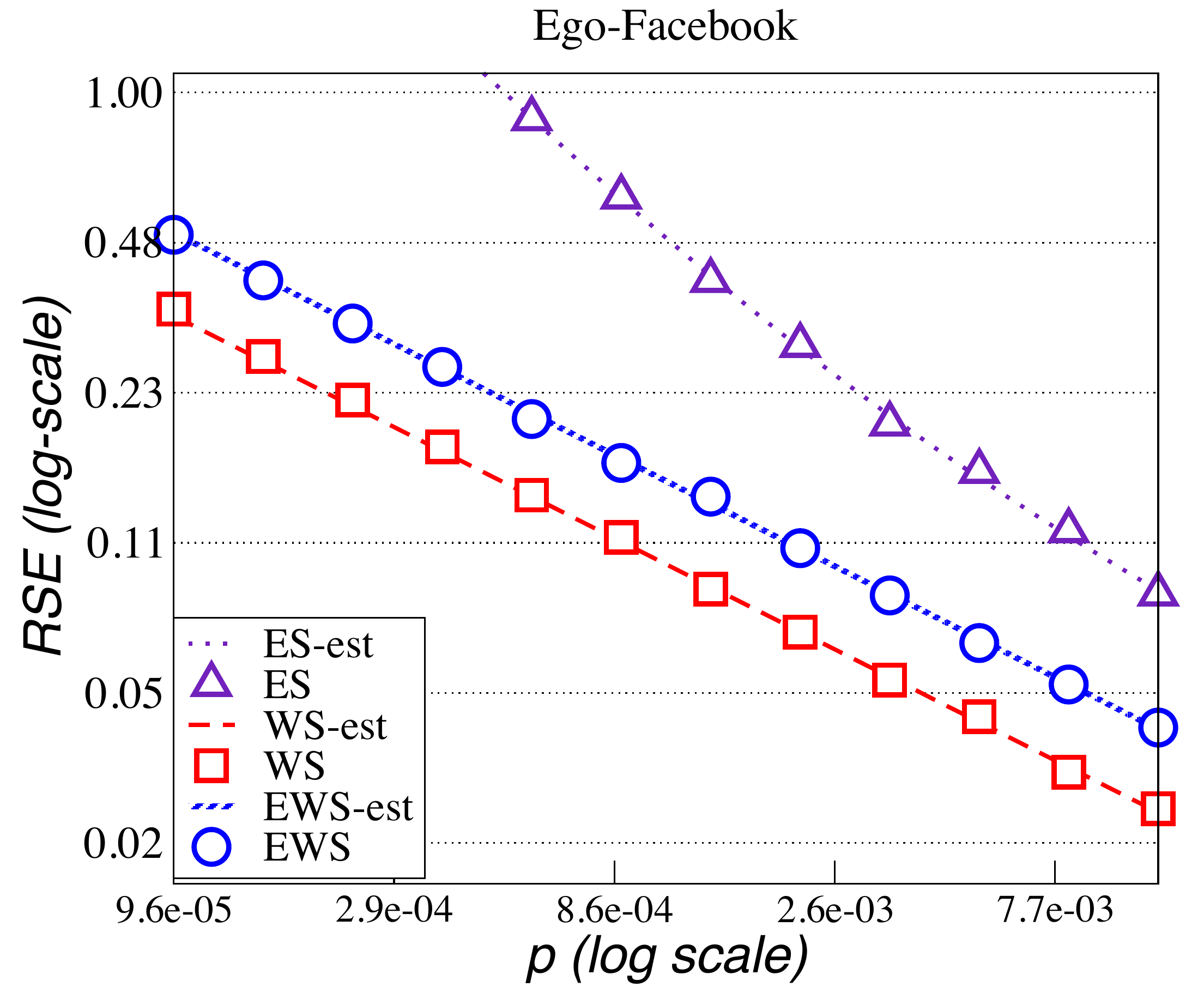}
        \caption{Ego-Facebook.}
        \label{fig:gull}
    \end{subfigure}
    \begin{subfigure}[b]{0.245\textwidth}
        \includegraphics[width=\textwidth]{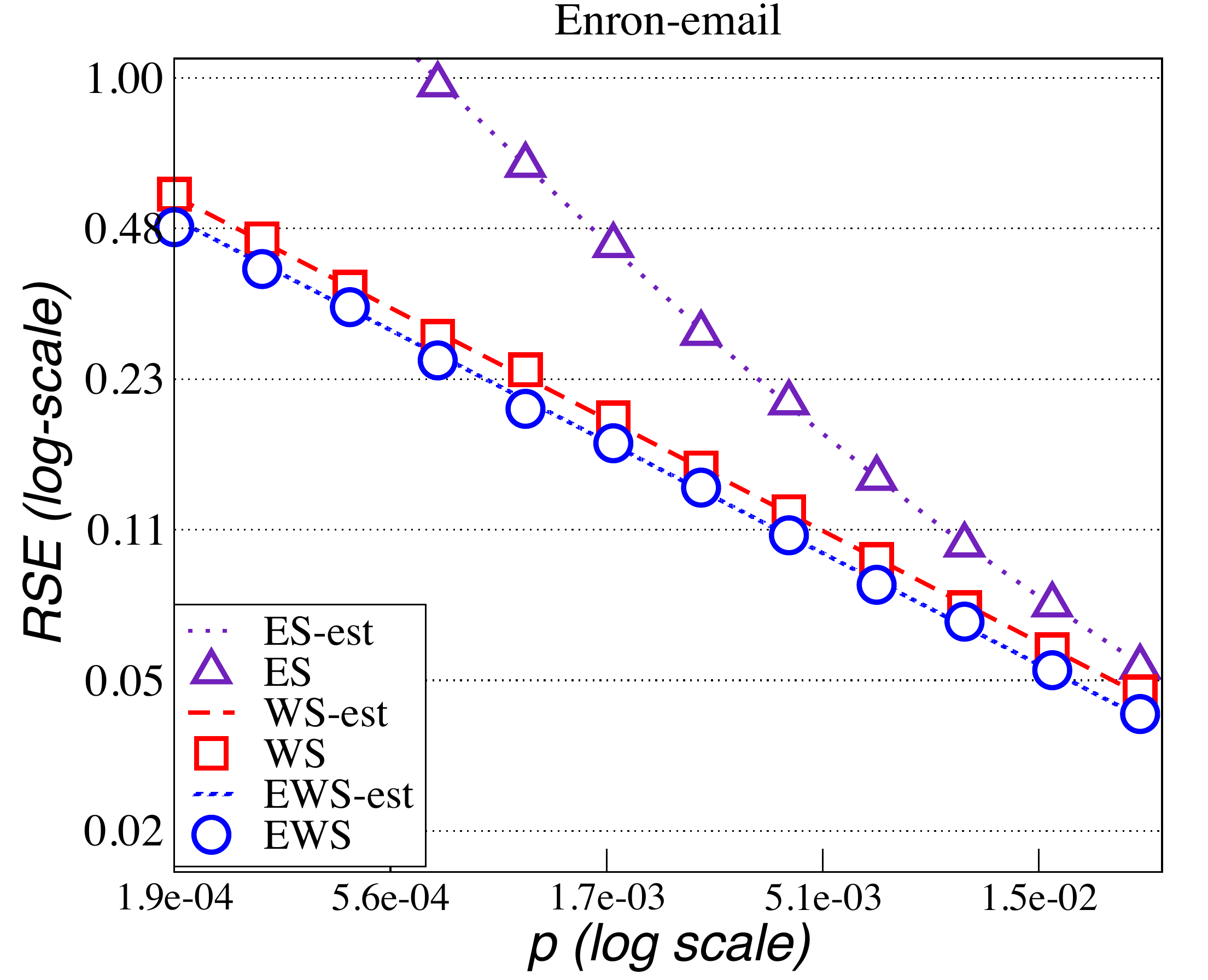}
        \caption{Enron-email.}
        \label{fig:gull}
    \end{subfigure}
    \begin{subfigure}[b]{0.245\textwidth}
        \includegraphics[width=\textwidth]{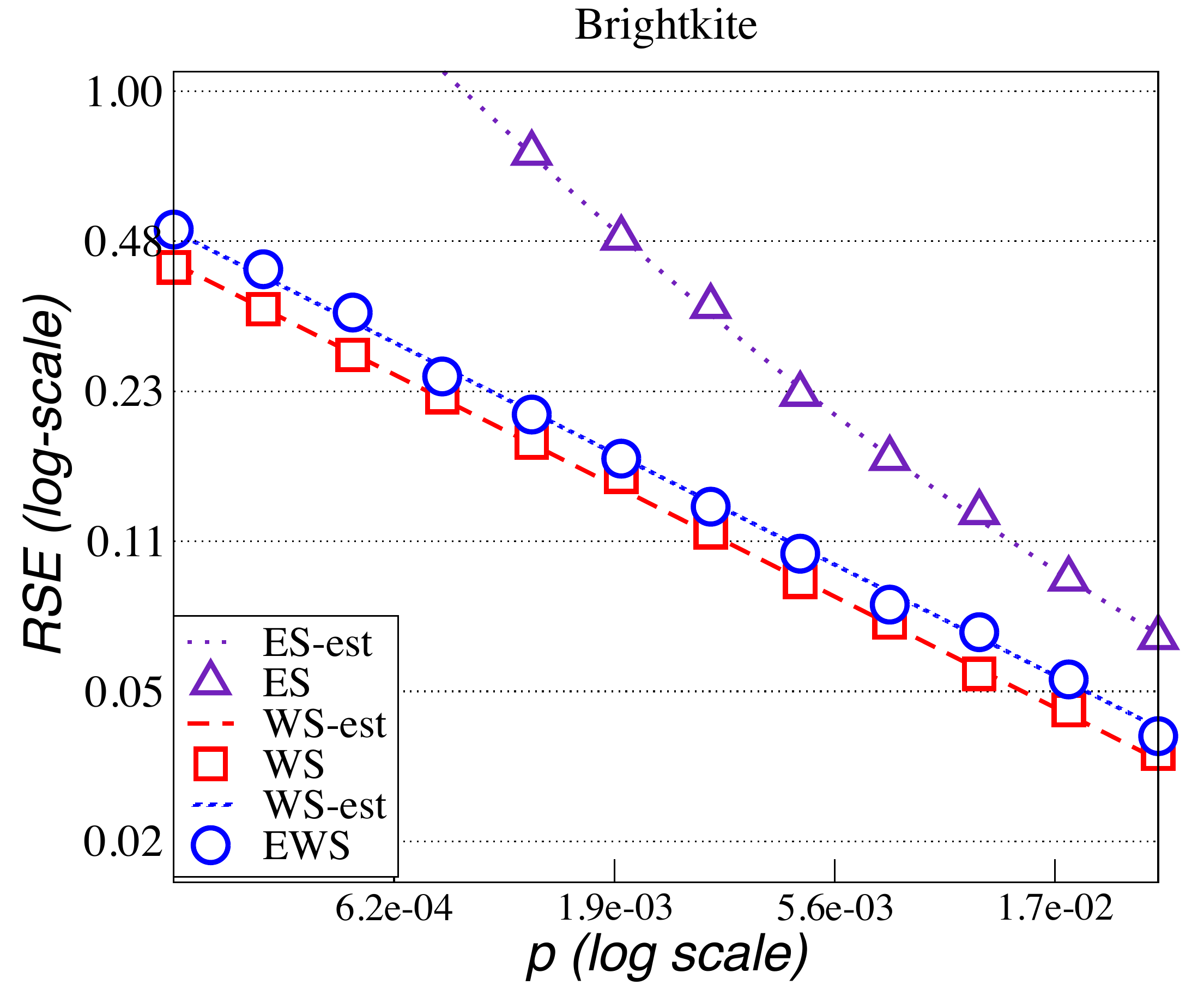}
        \caption{Brightkite.}
        \label{fig:gull}
    \end{subfigure}
    \begin{subfigure}[b]{0.245\textwidth}
        \includegraphics[width=\textwidth]{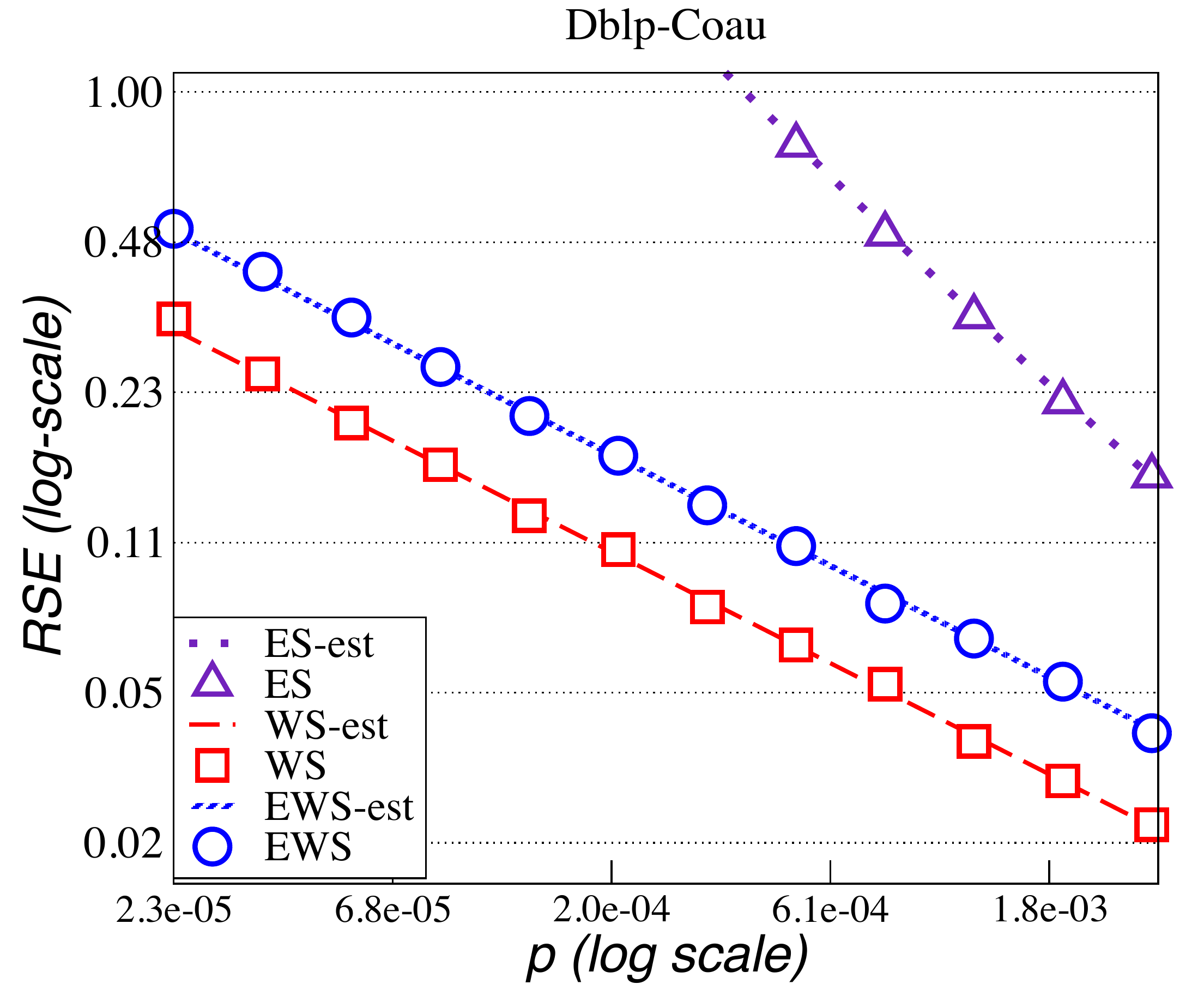}
        \caption{Dblp-Coau.}
        \label{fig:gull}
    \end{subfigure}
    \begin{subfigure}[b]{0.245\textwidth}
        \includegraphics[width=\textwidth]{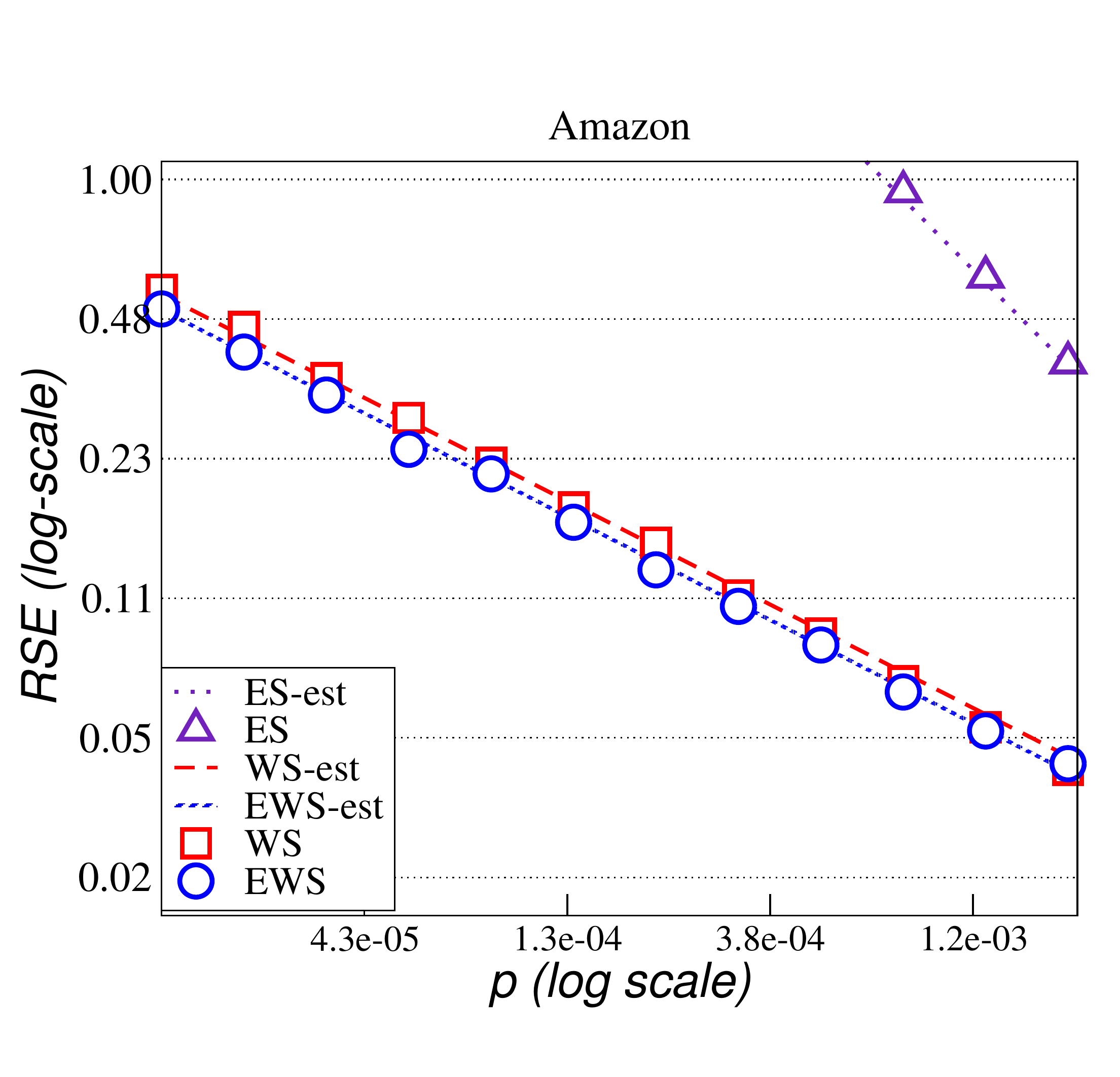}
        \caption{Amazon.}
        \label{fig:gull}
    \end{subfigure}
    \begin{subfigure}[b]{0.245\textwidth}
        \includegraphics[width=\textwidth]{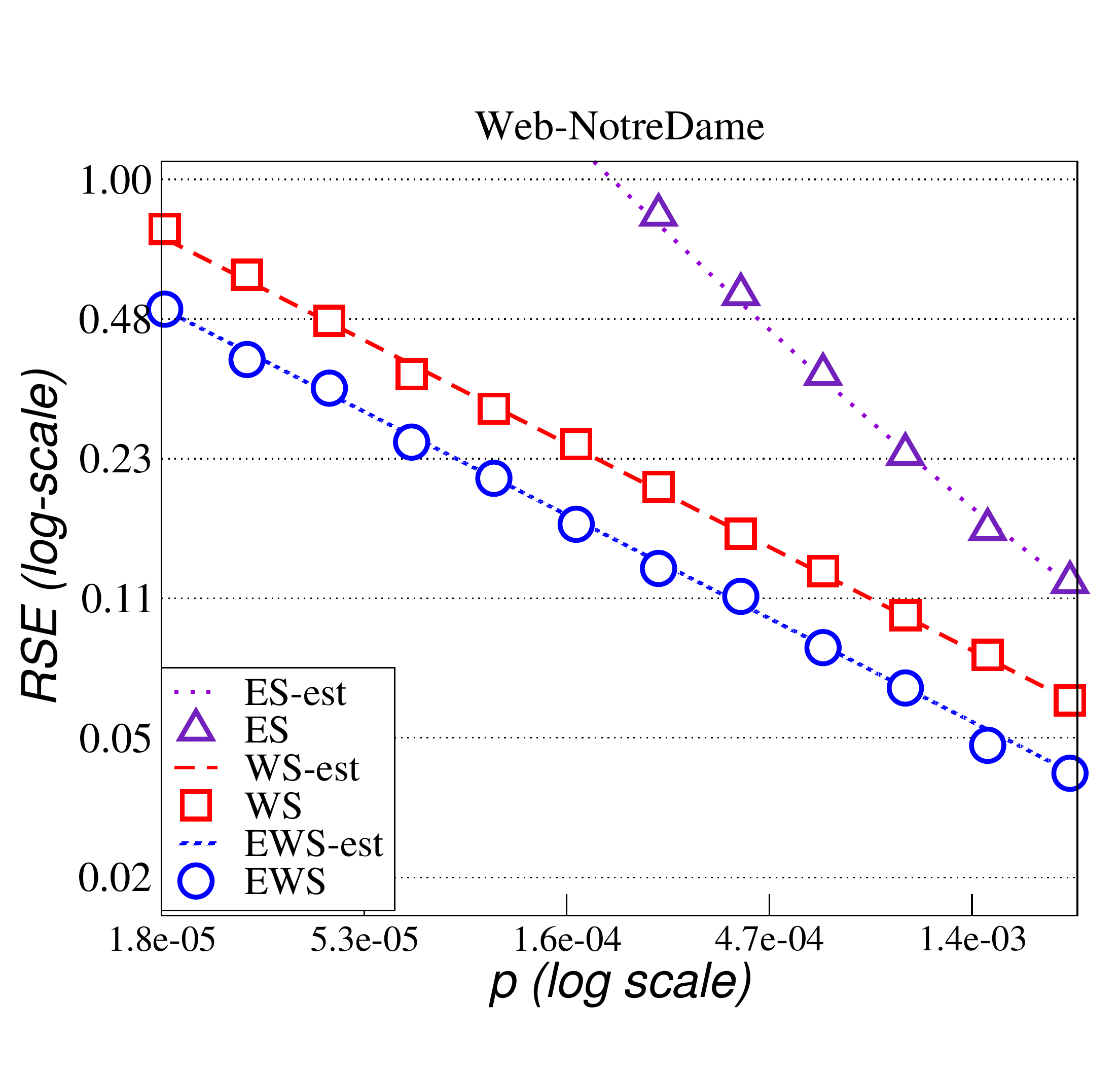}
        \caption{Web-NotreDame.}
        \label{fig:gull}
    \end{subfigure}
    \begin{subfigure}[b]{0.245\textwidth}
        \includegraphics[width=\textwidth]{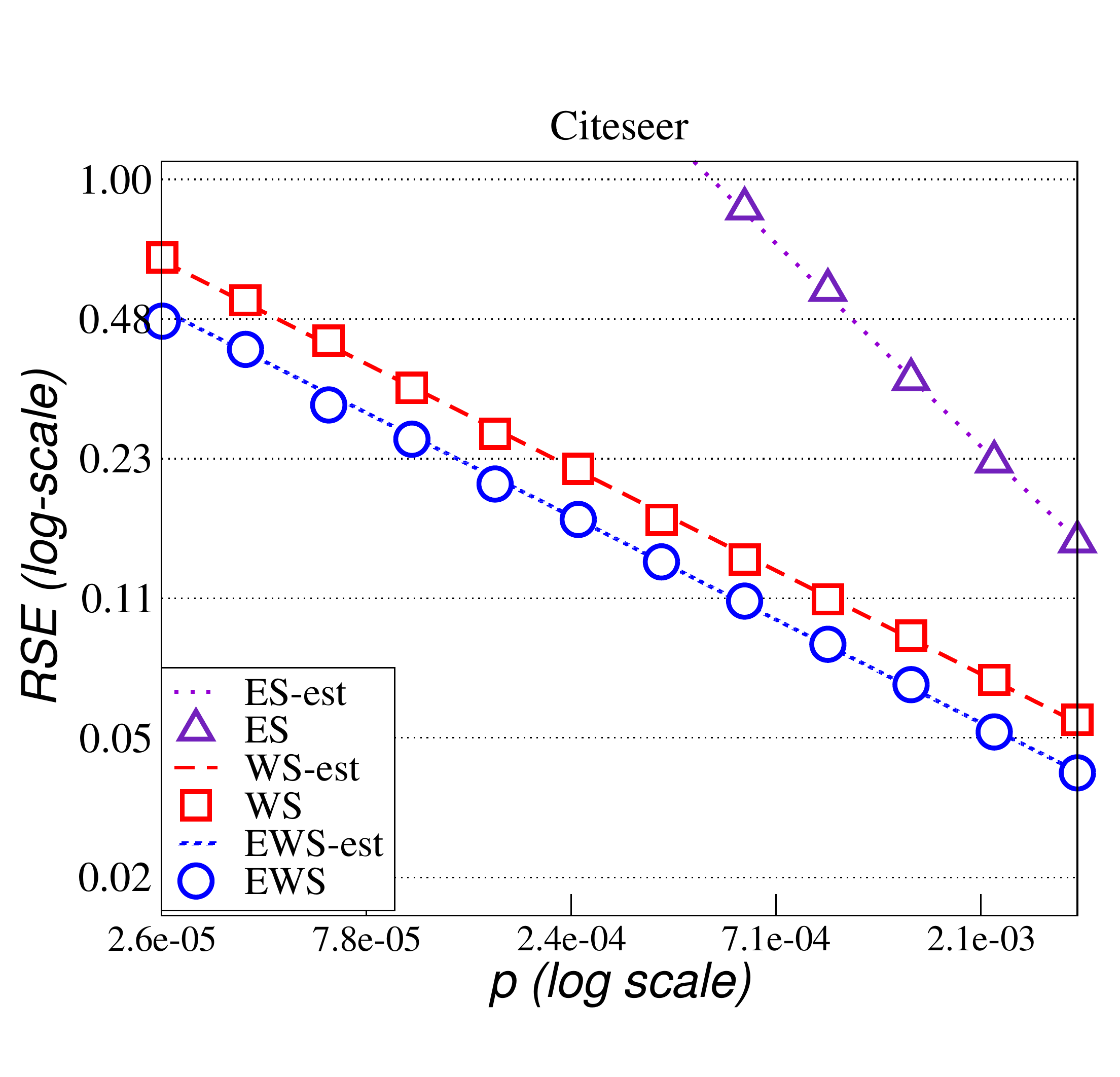}
        \caption{Citeseer.}
        \label{fig:gull}
    \end{subfigure}
    \begin{subfigure}[b]{0.245\textwidth}
        \includegraphics[width=\textwidth]{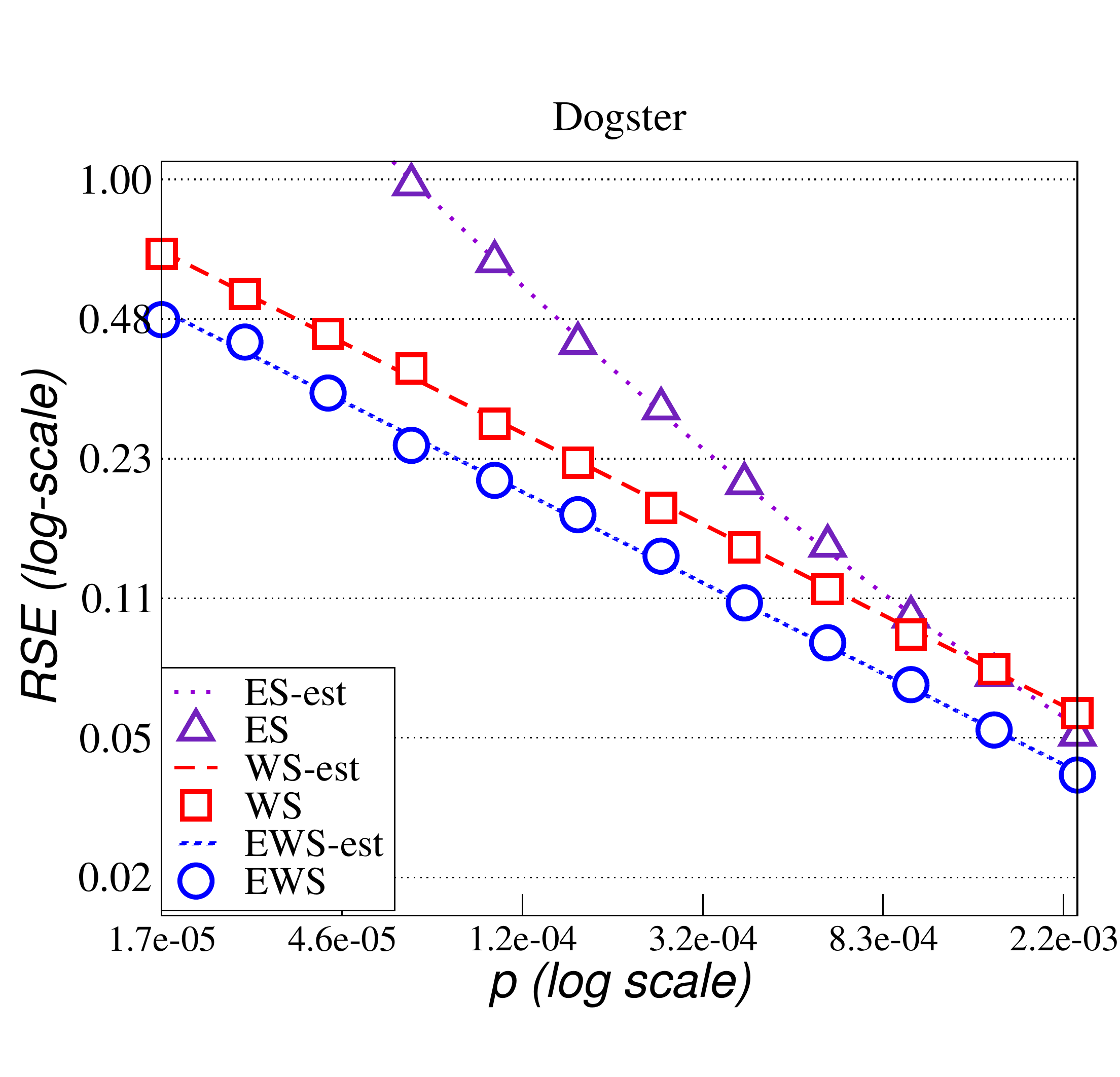}
        \caption{Dogster.}
        \label{fig:gull}
    \end{subfigure}
    \begin{subfigure}[b]{0.245\textwidth}
        \includegraphics[width=\textwidth]{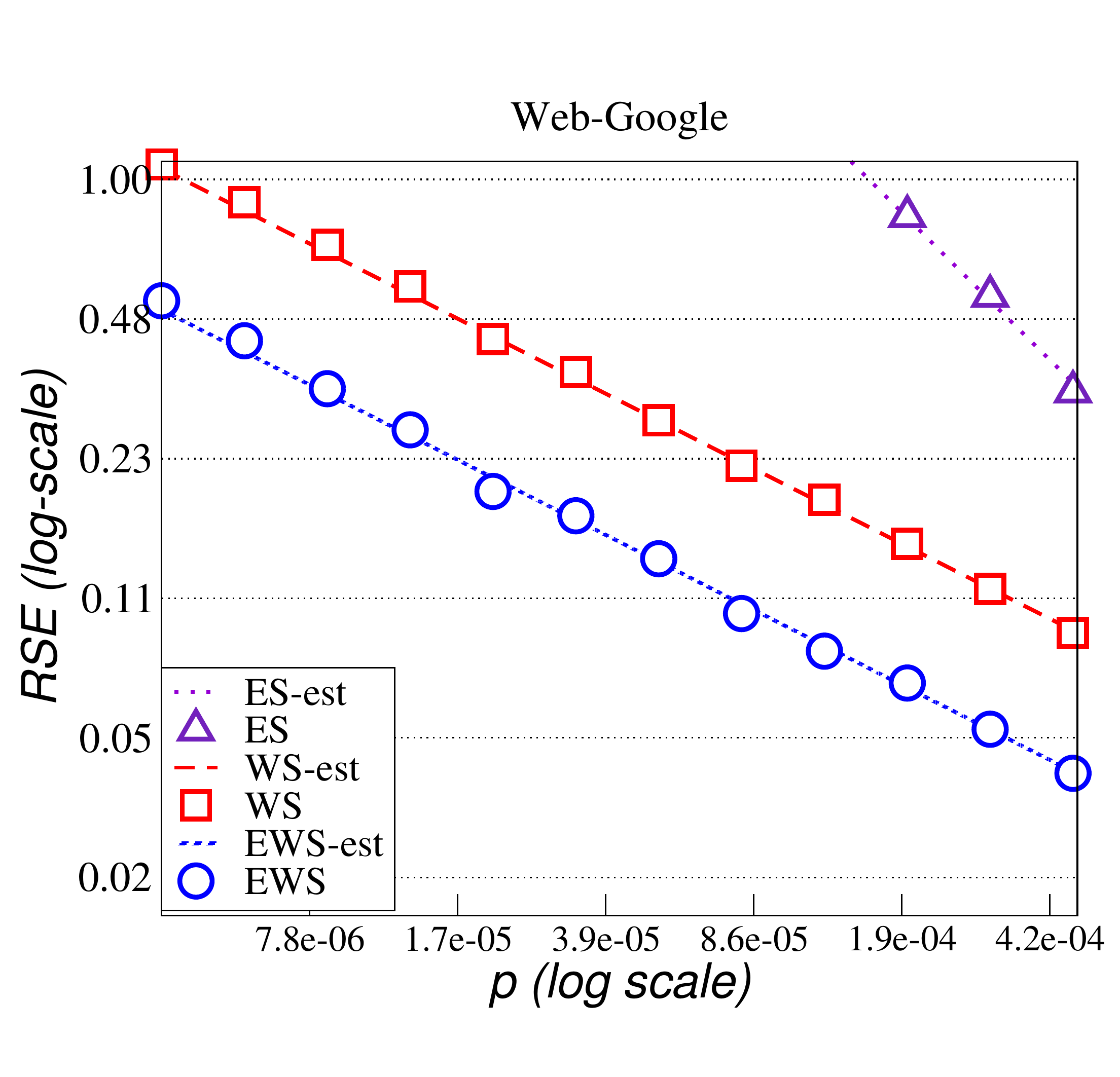}
        \caption{Web-Google.}
        \label{fig:gull}
    \end{subfigure}
    \begin{subfigure}[b]{0.245\textwidth}
        \includegraphics[width=\textwidth]{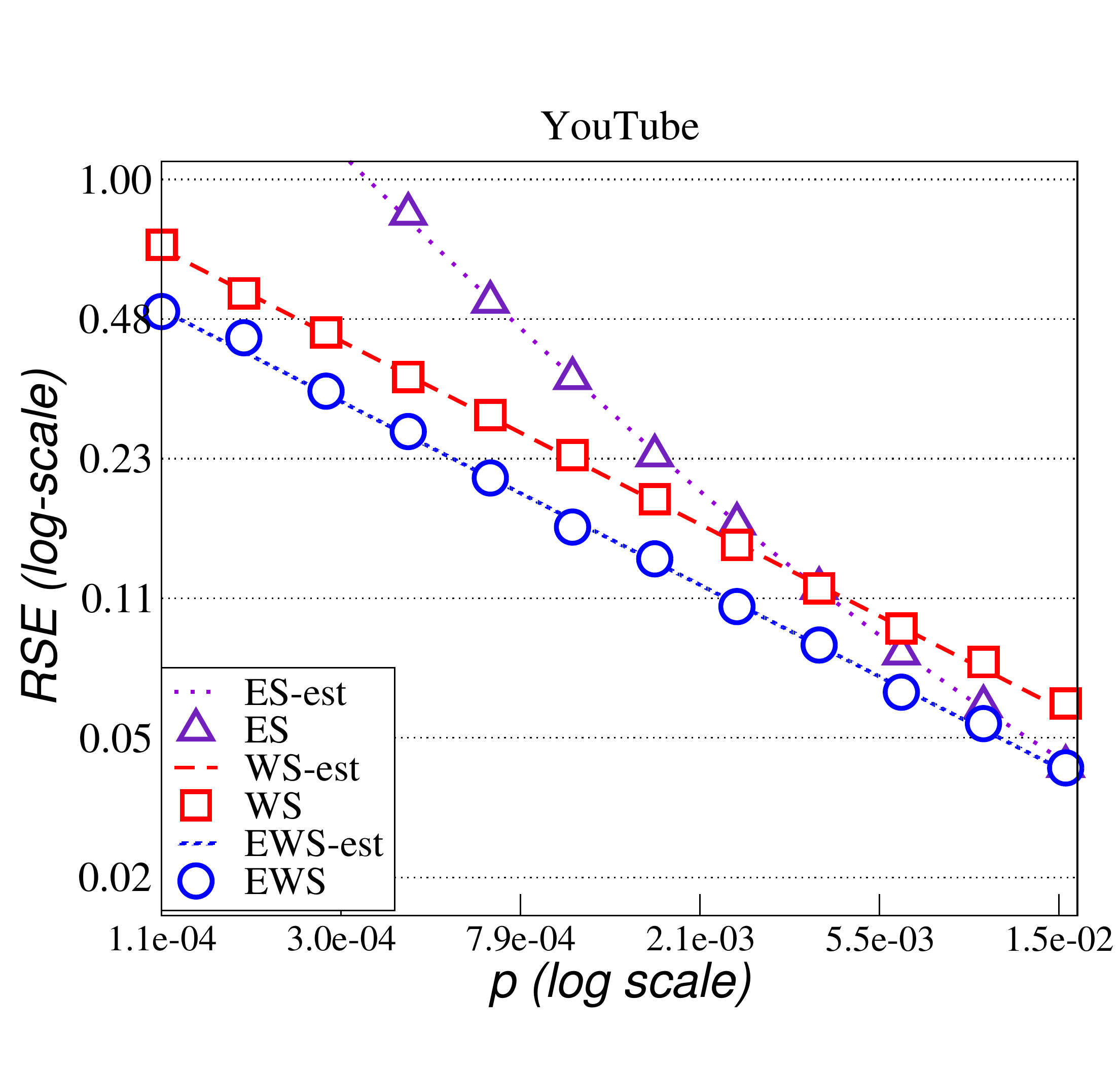}
        \caption{YouTube.}
        \label{fig:gull}
    \end{subfigure}
    \begin{subfigure}[b]{0.245\textwidth}
        \includegraphics[width=\textwidth]{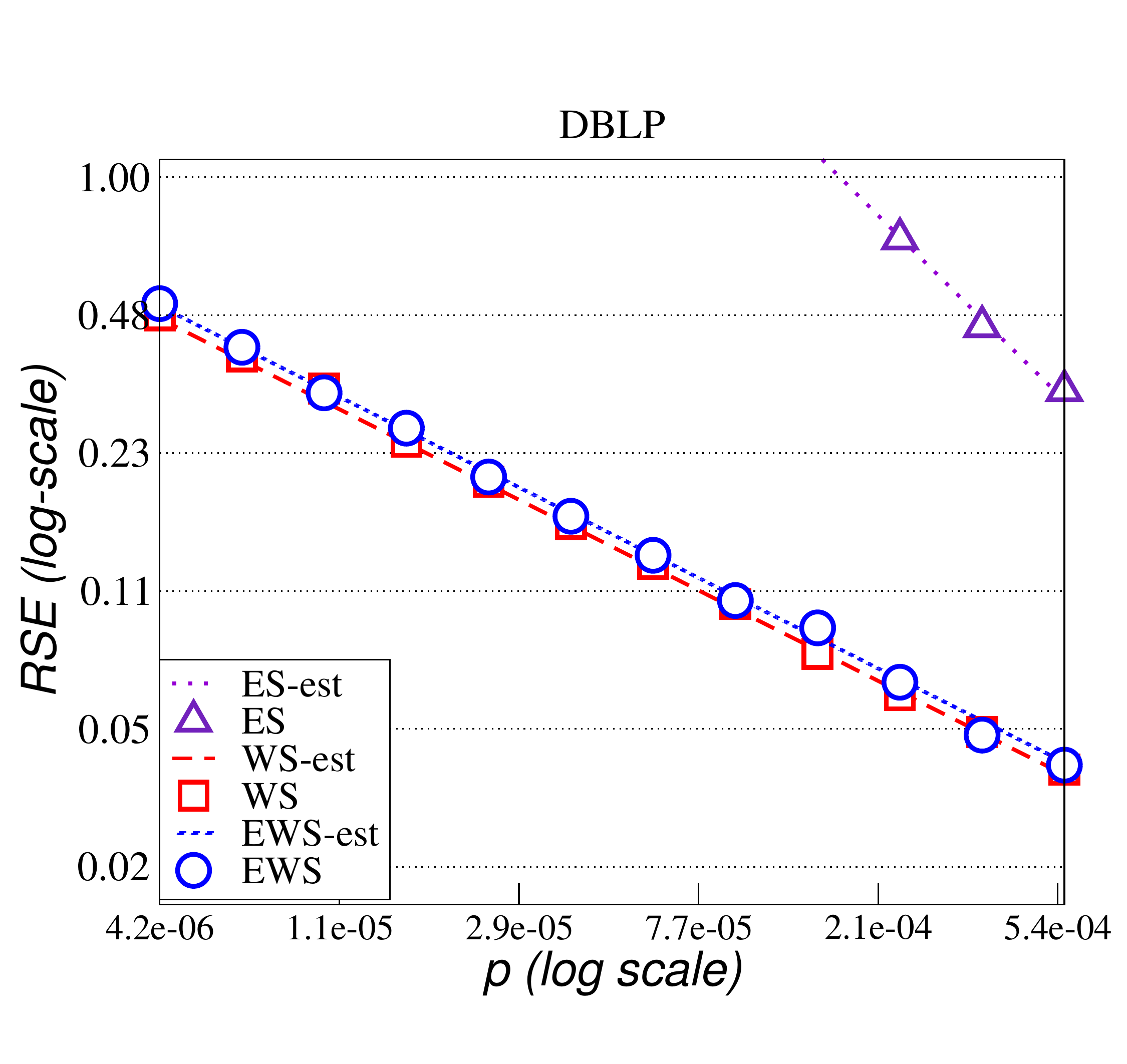}
        \caption{DBLP.}
        \label{fig:gull}
    \end{subfigure}
    \begin{subfigure}[b]{0.245\textwidth}
        \includegraphics[width=\textwidth]{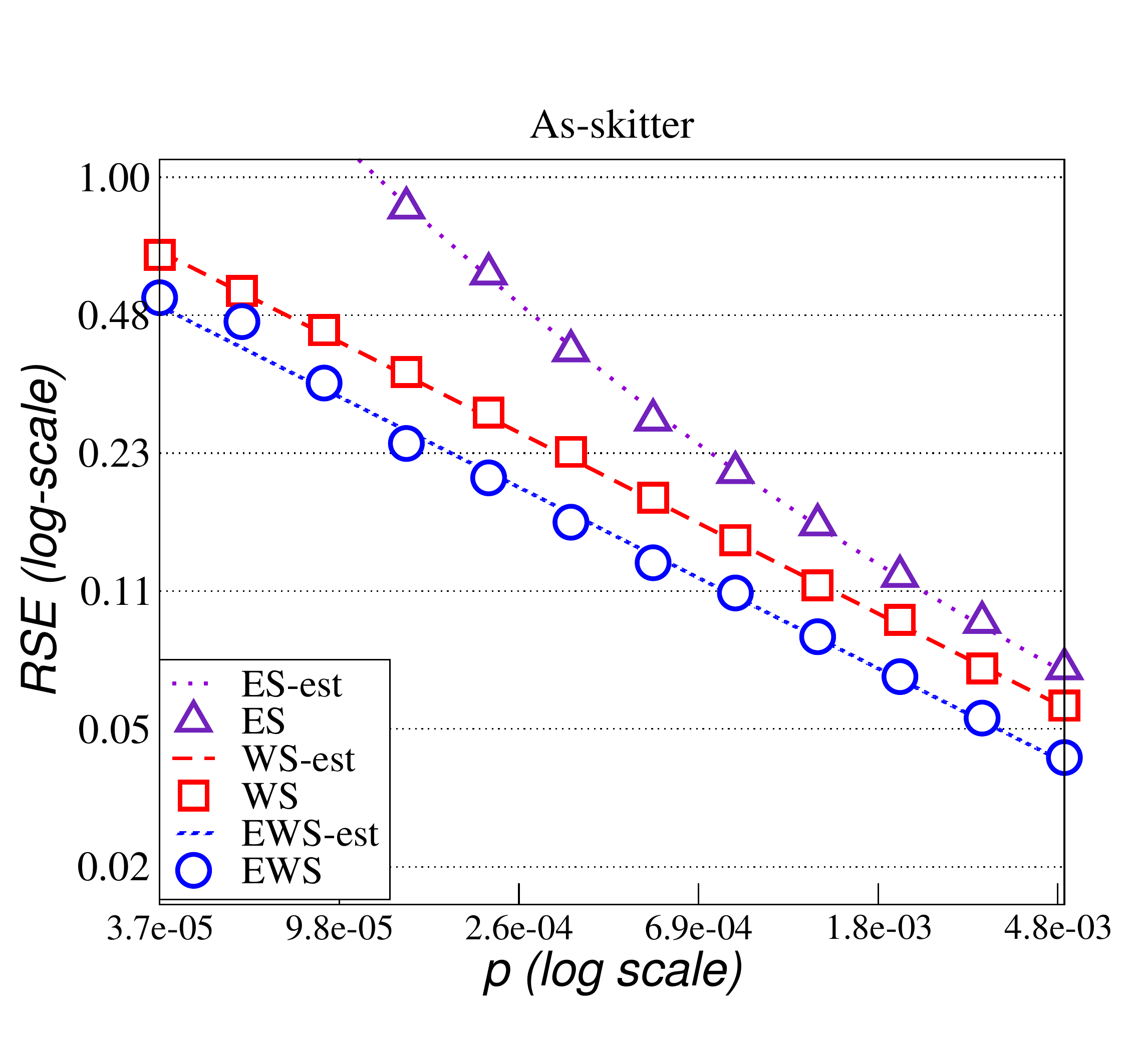}
        \caption{As-skitter.}
        \label{fig:gull}
    \end{subfigure}
      \begin{subfigure}[b]{0.245\textwidth}
        \includegraphics[width=\textwidth]{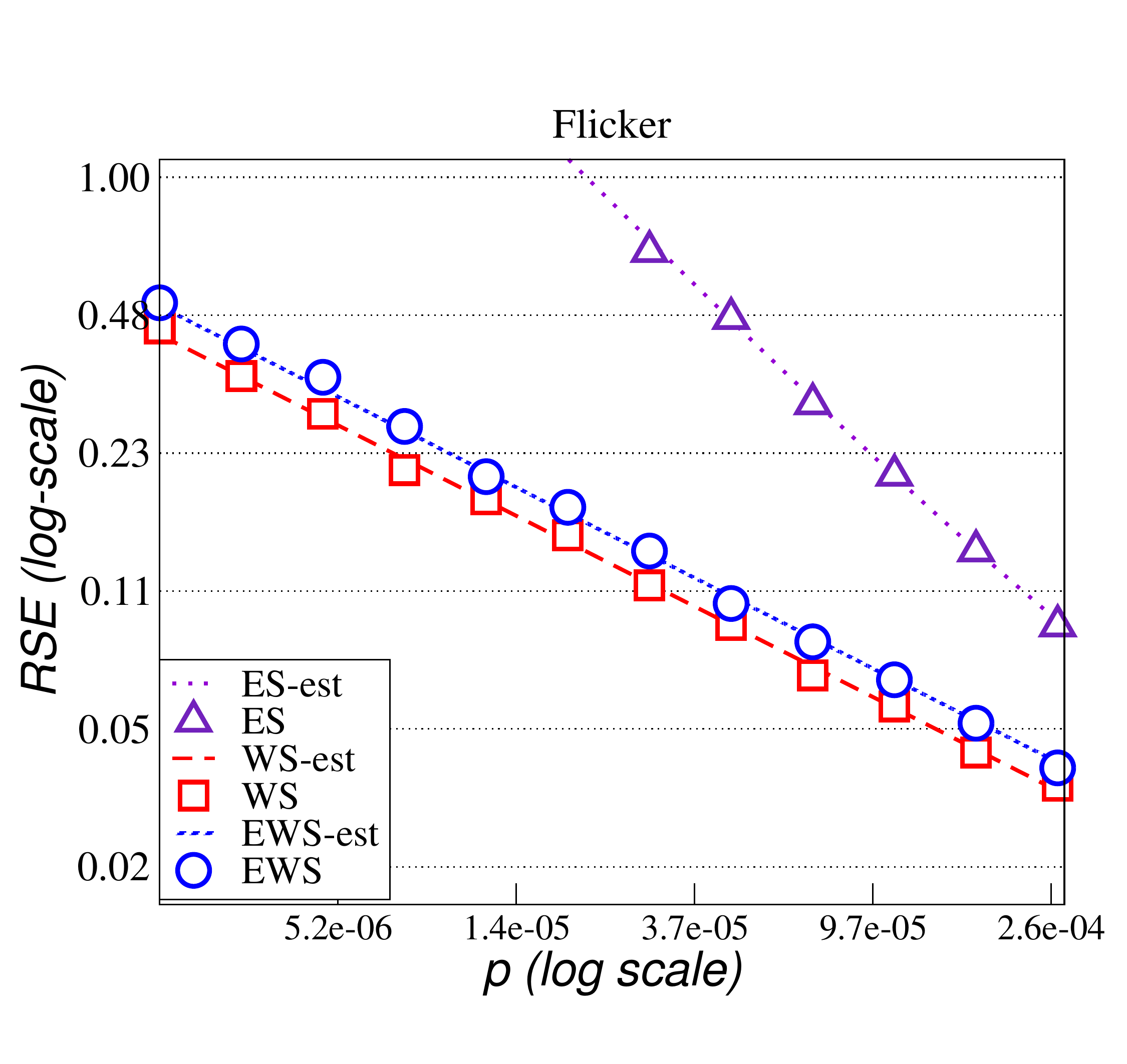}
        \caption{Flicker.}
        \label{fig:gull}
    \end{subfigure}
    \begin{subfigure}[b]{0.245\textwidth}
        \includegraphics[width=\textwidth]{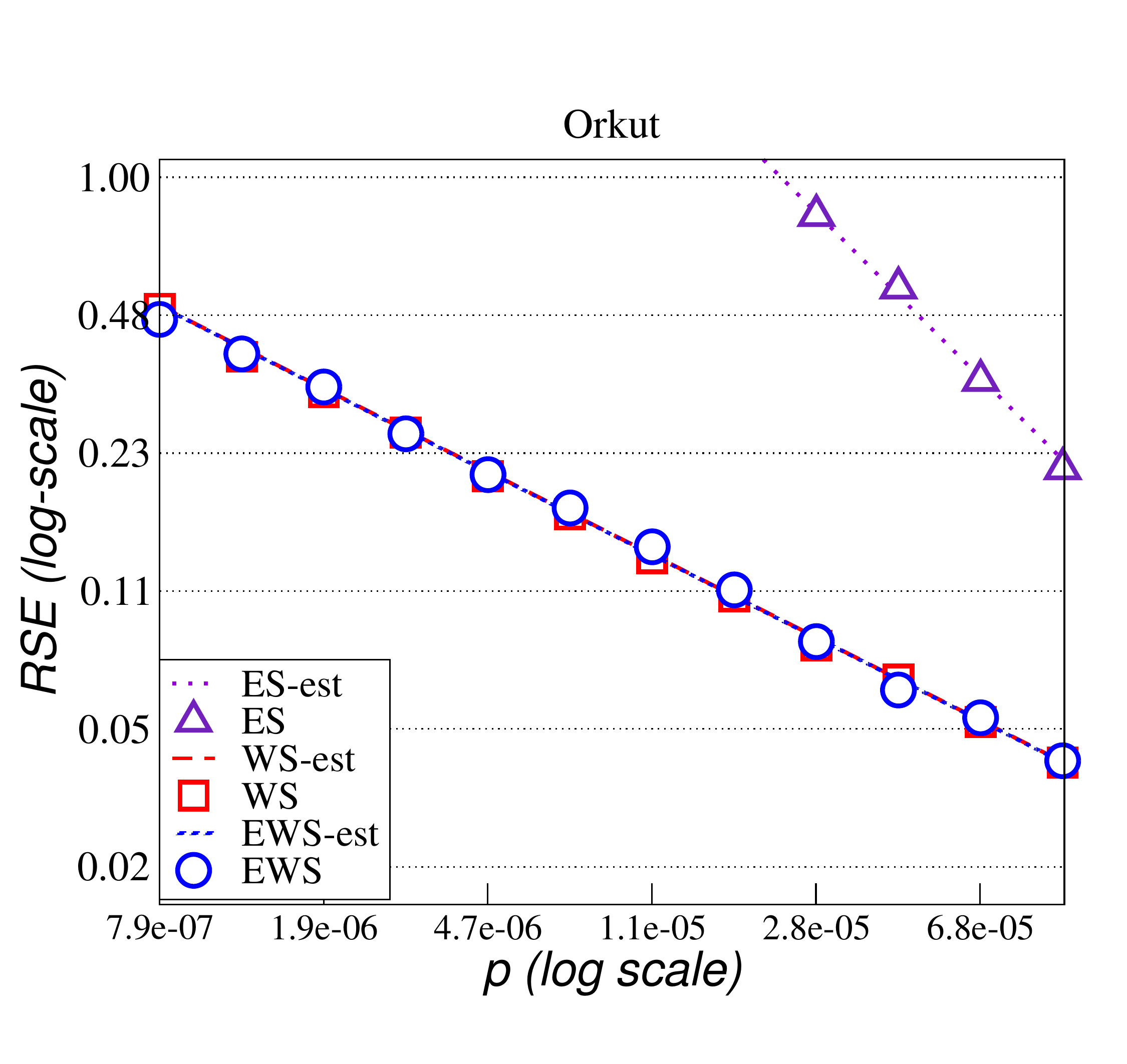}
        \caption{Orkut.}
        \label{fig:gull}
    \end{subfigure}
    \begin{subfigure}[b]{0.245\textwidth}
        \includegraphics[width=\textwidth]{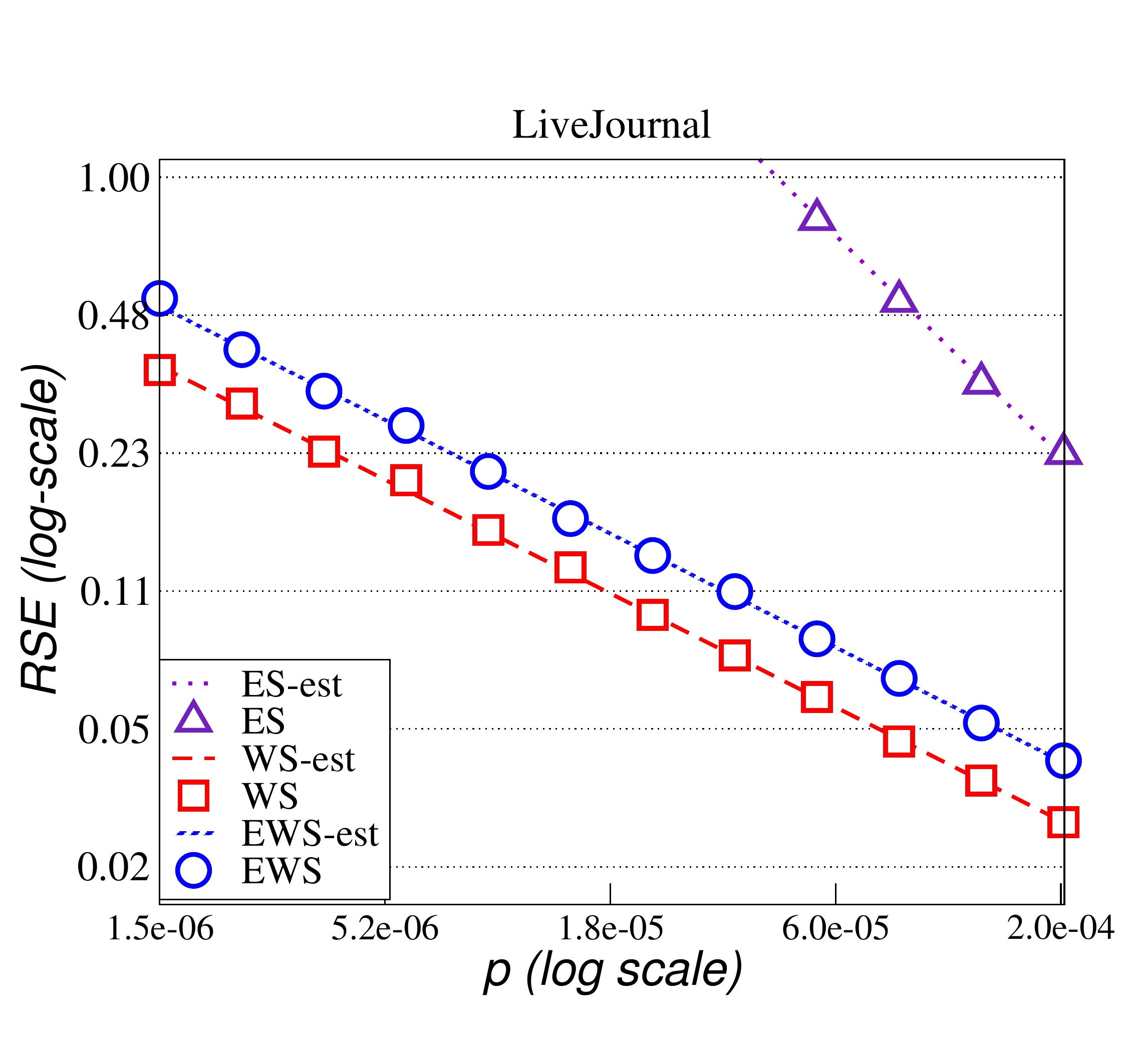}
        \caption{LiveJournal.}
        \label{fig:gull}
    \end{subfigure}
    \begin{subfigure}[b]{0.245\textwidth}
        \includegraphics[width=\textwidth]{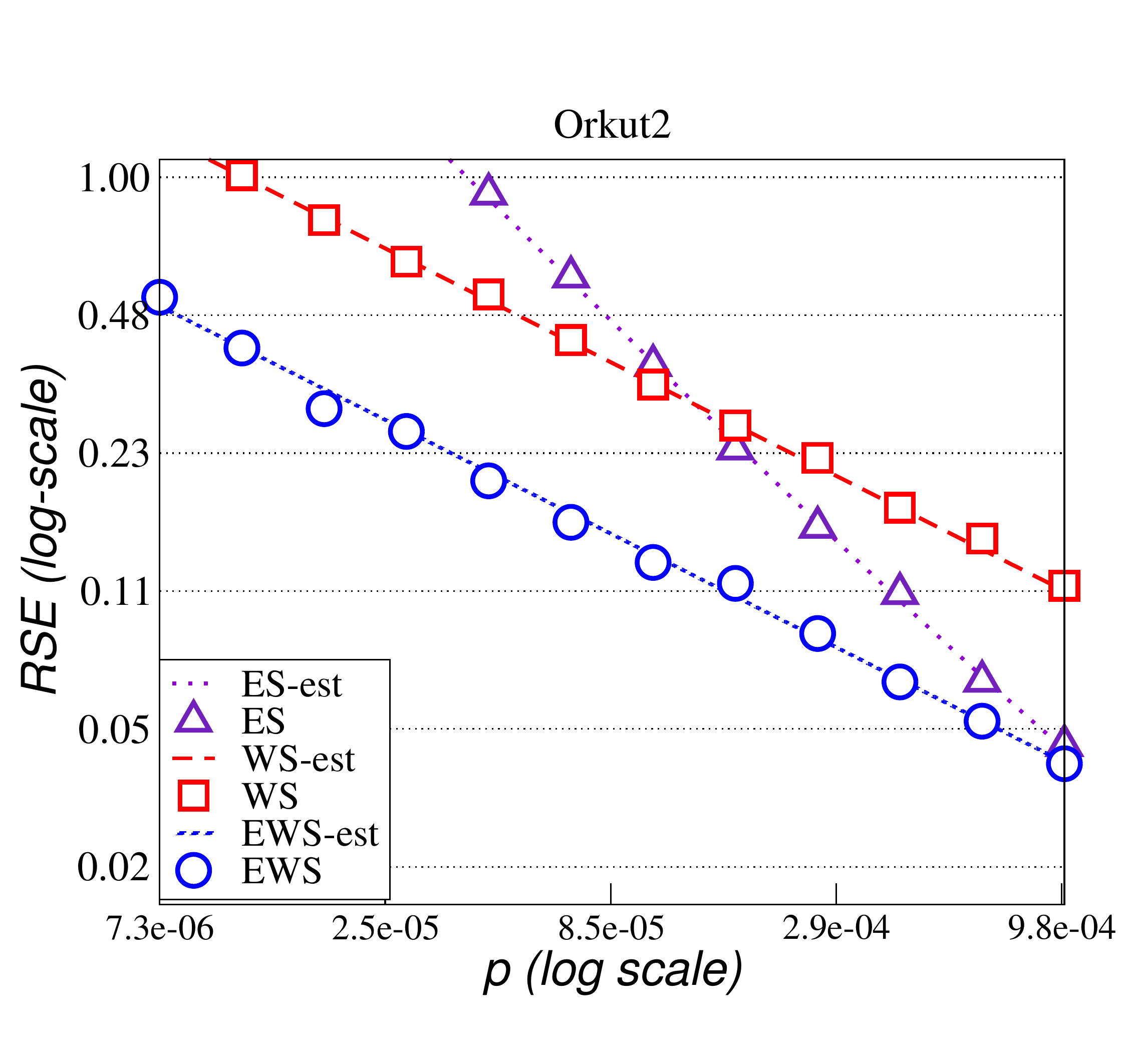}
        \caption{Orkut2.}
        \label{fig:gull}
    \end{subfigure}
    \begin{subfigure}[b]{0.245\textwidth}
        \includegraphics[width=\textwidth]{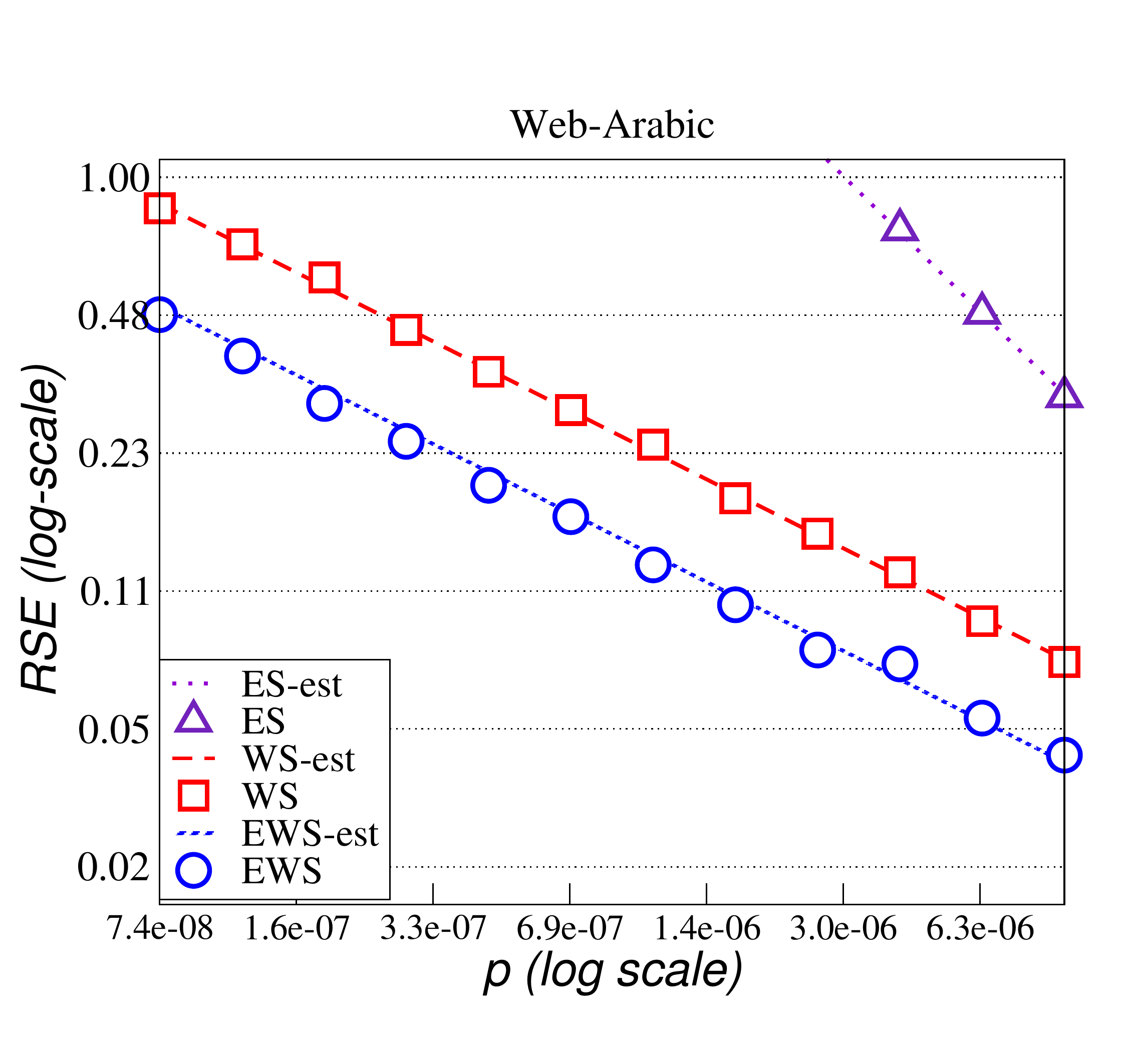}
        \caption{Web-Arabic.}
        \label{fig:gull}
    \end{subfigure}
    \begin{subfigure}[b]{0.245\textwidth}
        \includegraphics[width=\textwidth]{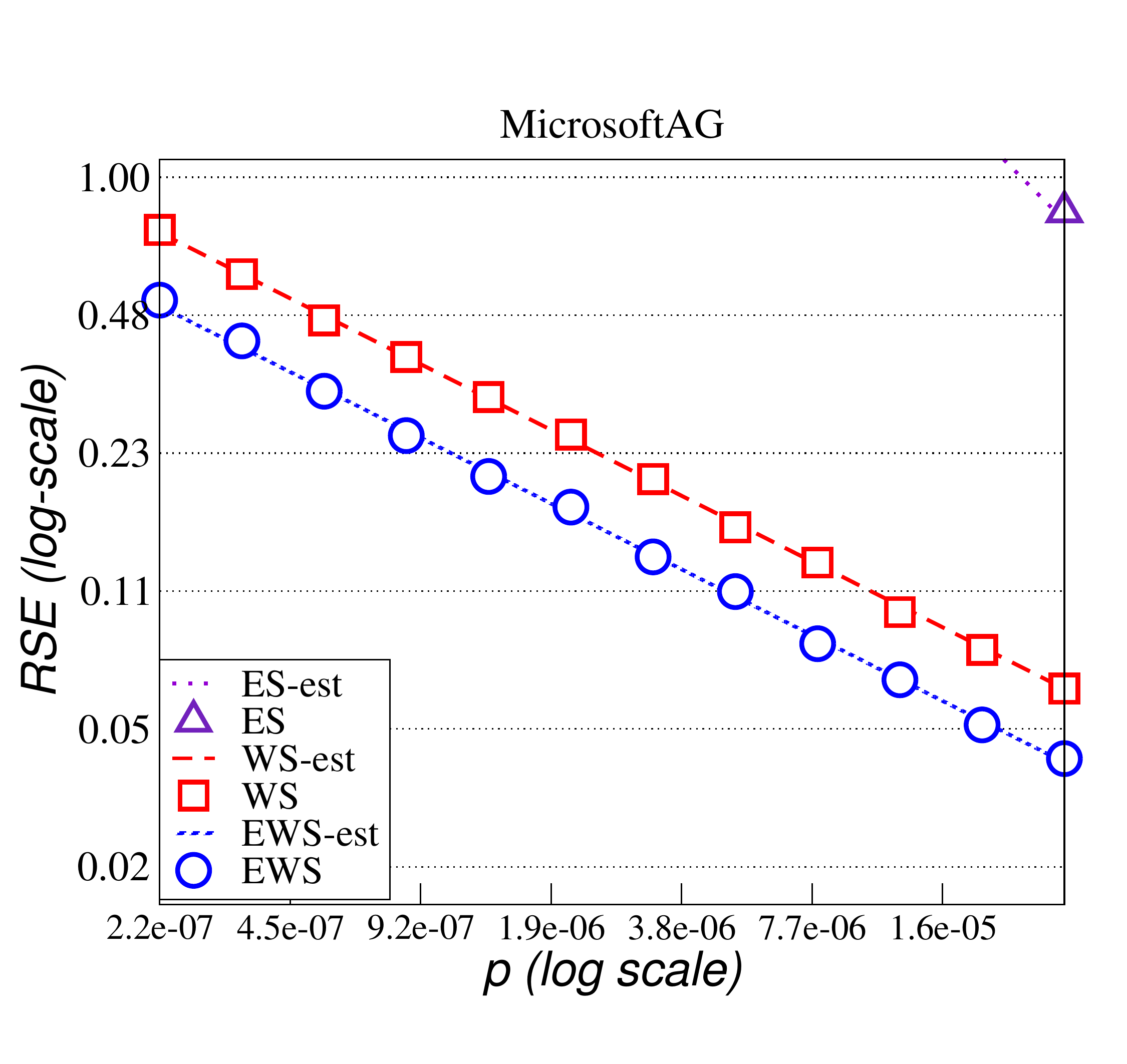}
        \caption{MicrosoftAG.}
        \label{fig:gull}
    \end{subfigure}
    \begin{subfigure}[b]{0.245\textwidth}
        \includegraphics[width=\textwidth]{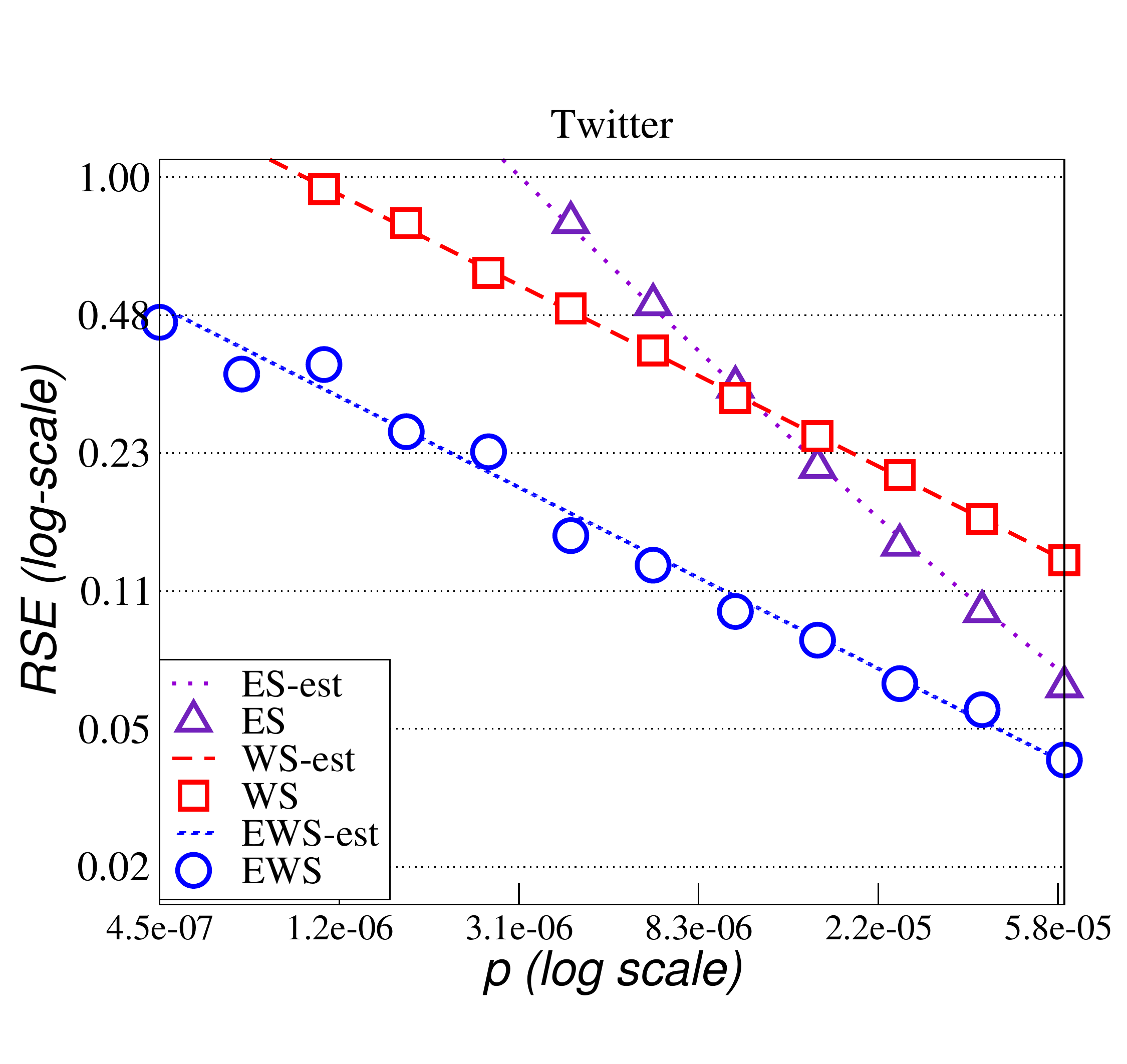}
        \caption{Twitter.}
        \label{fig:gull}
    \end{subfigure}
    \begin{subfigure}[b]{0.245\textwidth}
        \includegraphics[width=\textwidth]{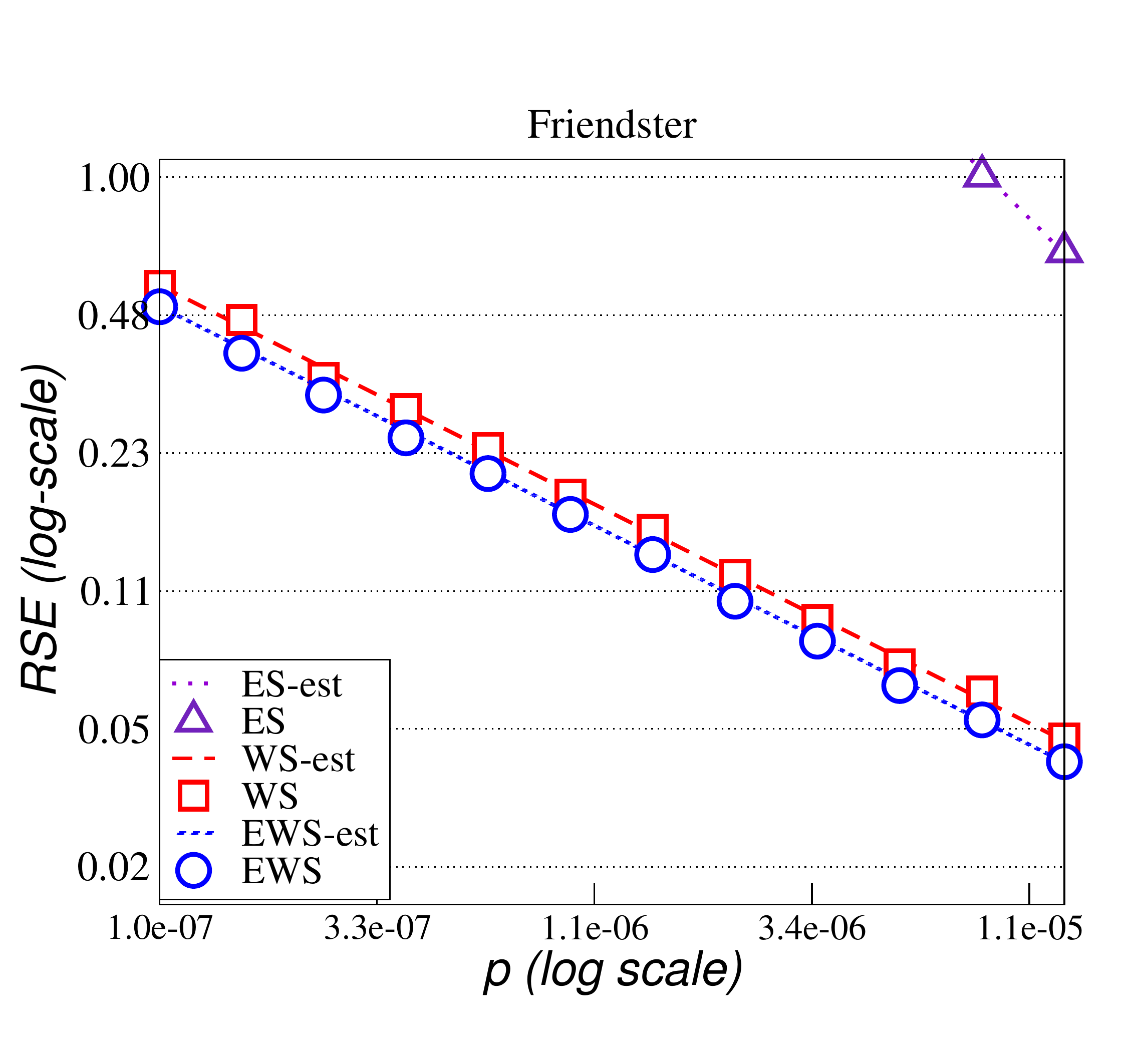}
        \caption{Friendster.}
        \label{fig:gull}
    \end{subfigure}
    \caption{$RSE$ vs $p$ for the 20 datasets.}
    \label{fig:rse}
\end{figure*}

\begin{table}
\caption{Sample sizes to achieve RSE=0.05.}
\label{tbl:datasets-sample}
\centering
\begin{tabular}{l r r r r c}
Dataset & $m$ & $ES$ & $WS$ & $EWS$ & $\frac{WS}{EWS}$ 
\vspace{1mm}
\\
\hline  
Ego-Facebook     &        88K &     2978  &     370  &    843 & 0.44 \\
Enron-email      &       183K &     5619  &    4288  &   3443 & {\bf 1.25} \\
Brightkite       &       214K &    10299  &    3217  &   4417 & 0.73 \\
Dblp-coauthor    &      1049K &    13005  &     905  &   2369 & 0.38 \\
Amazon           &       925K &    14460  &    1549  &   1316 & {\bf 1.18} \\
Web-NotreDame    &      1090K &     9459  &    4162  &   1948 & {\bf 2.14} \\
Citeseer         &      1736K &    20102  &    7660  &   4554 & {\bf 1.68} \\
Dogster          &      8543K &    23326  &   27631  &  14805 & {\bf 1.87} \\
Web-Google       &      4322K &    16556  &    6842  &   1525 & {\bf 4.49} \\
Youtube          &      2987K &    38692  &   63923  &  33586 & {\bf 1.90} \\
DBLP             &      5362K &    20763  &    1949  &   2228 & 0.87 \\
As-skitter       &     11095K &   103336  &   73852  &  40950 & {\bf 1.80} \\
Flicker          &     22838K &    13332  &    3315  &   4507 & 0.74 \\
Orkut            &    117185K &    57679  &    9293  &   9244 & {\bf 1.01} \\
LiveJournal      &     34681K &    38703  &    2791  &   5314 & 0.52 \\
Orkut2           &    327036K &   296516  & 1519667  & 240184 & {\bf 6.33} \\
Web-Arabic       &    553903K &    42394  &   12363  &   4098 & {\bf 3.02} \\
MicrosoftAG      &    528463K &   259340  &   26135  &  11836 & {\bf 2.21} \\
Twitter          &   1202513K &   111704  &  472194  &  53584 & {\bf 8.81} \\
Friendster       &   1806067K &   326237  &   22621  &  17976 & {\bf 1.26} \\
\hline
\end{tabular}
\end{table}

To provide a deeper understanding of the difference between the three
algorithms we compare, we also present 
Table~\ref{tbl:datasets-sample}, where we fix the RSE to 0.05 (95\%
confidence value) and compare the number of sampled entities that
allows all approaches to achieve this given RSE value using the
theoretic approximations from
Equations~\ref{eqn:rse-tau},~\ref{eqn:rse-omega},
and~\ref{eqn:rse-rho}. As seen in
the table, in all datasets either EWS or WS provides the lowest
sampling size. Hence, in the last column of the table we provide the
ratio of the sampling sizes for WS and EWS. Similar to the results
provided in Figure~\ref{fig:rse}, in 14 datasets the $\frac{WS}{EWS}$
value is higher than one, indicating that EWS outperforms WS, and for
6 datasets it is lower. 
Furthermore, because the RSE values of both approaches depend linearly
on the sampling ratio, the $\frac{WS}{EWS}$ ratio in
Table~\ref{tbl:datasets-sample} would remain the same if some
other fixed RSE value was used instead.

Noticeably, the datasets that WS outperforms EWS are datasets with
high global clustering coefficient values (e.g. $C \geq 0.1$). This is
expected as in these datasets the probability of observing
closed wedges when sampling wedges increase significantly and the
advantage of EWS over WS is lost.
However, EWS outperforms WS for all graphs with low global clustering
coefficient values ($C < 0.1$) as the performance of WS is directly
tied to this coefficient. We also note that advantage of EWS over WS becomes more prominent as the sizes of graphs increase. 

Complexity of EWS, ES, and WS are all $O(m)$ and practical runtimes of these algorithms with low sampling ratios over very large graphs are very fast. As an example, on the Twitter dataset, with a sampling ratio of $p \approx 0.00006$,
EWS, ES, and WS takes 1.12, 1.11, and 0.26 seconds. 

\section{Conclusion}
\label{sec:conclusion}

We proposed an edge-based wedge sampling approach for estimating the number of triangles in very large power-law degree graphs. Our approach combines the benefits of edge and wedge sampling to offer highly accurate estimations even for very large sparse graphs and for very low sampling ratios. Furthermore, it does not require any preprocessing to be performed over the graphs. Through analysis conducted over graphs modeling large-scale real-world networks, we theoretically and empirically show that our approach offers highly confident estimations and up to eight times sampling size reduction over the state-of-the-art alternatives even when the sampling ratio is low.

%


\balance
\bibliographystyle{IEEEtran}



%
%
%

\end{NoHyper}
\end{document}